\documentclass[conference]{IEEEtran}
\IEEEoverridecommandlockouts
\usepackage{cite}
\usepackage{amsmath,amssymb,amsfonts}
\usepackage{graphicx}
\usepackage{textcomp}
\usepackage[table]{xcolor}
\usepackage{soul}
\usepackage[most]{tcolorbox}
\usepackage{amsthm}
\usepackage{algorithm}
\usepackage{algpseudocode}
\newtheorem{theorem}{Theorem}
\usepackage{subcaption}
\usepackage{booktabs,array,multirow,adjustbox}
\definecolor{oursblue}{RGB}{218,245,248}

\def\BibTeX{{\rm B\kern-.05em{\sc i\kern-.025em b}\kern-.08em
    T\kern-.1667em\lower.7ex\hbox{E}\kern-.125emX}}
\begin{document}

\title{FuseRAG: Fusing Structure and Evidence into Inference for Cost-Efficient Graph-based RAG
}

\author{
\IEEEauthorblockN{
Tingqi Wang,
Yingshu Li,
Zhipeng Cai\textsuperscript{*}
}
\IEEEauthorblockA{
\textit{Department of Computer Science, Georgia State University} \\
Atlanta, GA, United States \\
\{twang33, yili, zcai\}@gsu.edu
}
}
\maketitle

\begin{abstract}
Graph-based retrieval-augmented generation (RAG) exploits relationships among evidence to support multi-hop question answering. However, conveying retrieved graphs through non-unique linear textual contexts to Large Language Models (LLMs) leads to graph-structure degradation and inefficient KV-cache reuse.
Meanwhile, graph-structure order optimization is already NP-hard, and the additional KV-cache reuse objective further complicates the search by introducing potentially conflicting orders.
To tackle these issues, we propose FuseRAG, a cost-efficient graph-based RAG framework that shifts most retrieved evidence from explicit textual contexts to latent representations. Our main idea is to separate reusable evidence content from query-aware structural interactions in latent space, while retaining only a few explicit facts as semantic anchors. The resulting latent representations are injected into decoder self-attention, where they reshape attention over evidence and thereby influence hidden states and next-token probabilities.
This design preserves graph structure through latent interactions, enables evidence reuse beyond exact token-prefix matching.
Experiments on three multi-hop question-answering datasets with Llama-3-8B and Qwen3.5-9B show that FuseRAG improves answer quality while achieving up to 2.7$\times$ and 2.2$\times$ TTFT speedups over the fastest graph-based baseline, respectively.
\end{abstract}


\section{Introduction}
Large Language Models (LLMs) \cite{chang2024survey, wu2024evolutionary} have demonstrated significant capabilities in language understanding, greatly advancing the development of Artificial General Intelligence (AGI).
Despite their considerable success, LLMs still generate hallucinations and wrong answers due to the lack of domain-specific knowledge in pretraining data \cite{lewis2020retrieval, li2024dawn, asai2024self, li2025neutronrag}.
Retrieval-Augmented Generation (RAG) \cite{fan2024survey} methods have been widely adopted to mitigate this limitation by integrating external knowledge into LLMs \cite{luo2026gfm}.
Among these approaches, graph-based RAG \cite{zhu2026graph, luo2024reasoning, indepth, tong2026galaxyrag, li2025framework} has attracted substantial attention.
It constructs a knowledge graph from source documents and retrieves relevant entities and relations as evidence to augment the LLM context.
Graph captures complex relationships among entities, facilitating multi-hop reasoning along relational paths.

\begin{figure}[t]
    \centering
    \includegraphics[width=1\columnwidth]{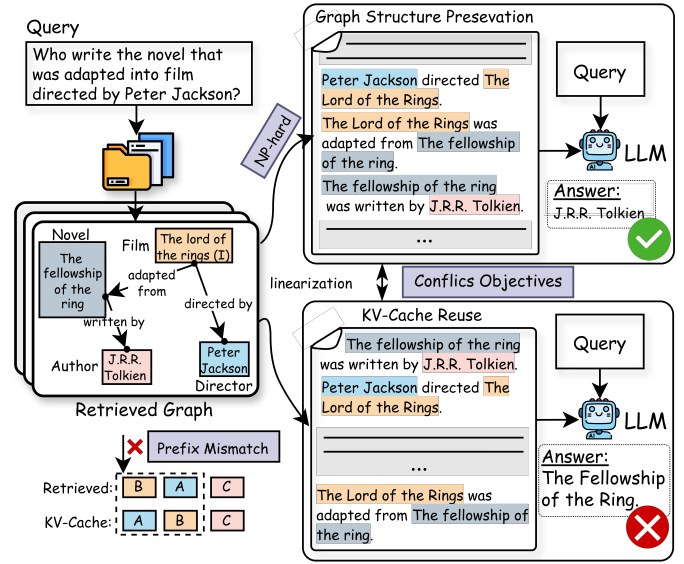}
    \caption{\textbf{The same retrieved graph with different evidence utilization.} Serialization orders aligned with the graph structure help LLMs utilize evidence effectively but reduce efficiency. However, orders optimized for KV-cache reuse weaken evidence utilization.}
    \label{fig:illustration}
\end{figure}

Although benefiting from graph-structure knowledge, 
existing graph-based RAG still suffers from limited expressiveness.
Specifically, before being fed into LLMs,
the retrieved graph-structured evidence must be serialized into linear textual contexts.
This mechanism poses a major challenge to graph-based RAG:
\textbf{Graph Structure Degradation.} 
After serialization, the relative positions of evidence implicitly describe the graph structure and influence evidence utilization \cite{liu2024lost}. 
However, a graph generally can be represented by multiple serialization orders, utilizing  evidence with different pattern.
As illustrated in Figure \ref{fig:illustration},
a query asking for the author of a novel adapted into a film directed by Peter Jackson. When the retrieved evidence is serialized along the reasoning path related to the query, adjacent sentences share entities and form a continuous chain (from Peter Jackson to the film, the novel, and finally J. R. R. Tolkien).
In contrast, factual entities separated by irrelevant content break the subgraph relationships and require LLMs to reason answers from multiple confusing information across distant contextual positions.
Consequently, reasoning paths become harder to trace, even with the same retrieved results.

Furthermore, the variability serialization of evidence \textbf{hinders the KV-cache reuse efficiency}. 
RAG typically introduces a large amount of external knowledge,
increasing the prefill cost.
As a common solution, KV-cache reuse accelerates LLM inference by caching the key-value states of previously processed tokens and reusing them.
However, these approaches are inherently designed for sequential text and assume identical token prefixes \cite{lu2022fantastically,zheng2024sglang}.
In graph-based RAG, inconsistent serialization orders undermine KV-cache reuse across equivalent evidence contexts and increasing inference cost.

The above issues motivate a method that effectively exploits graph structure while maintaining inference efficiency.
An intuitive idea is to optimize the serialization order for both graph-structure representation and KV-cache reuse.
Nevertheless, as shown in Section \ref{sec:p & m}, structural ordering is already NP-hard, and the additional cache-reuse objective further complicates the optimization.
The search cost grows rapidly with the graph size, making optimal serialization impractical for large-scale graph-based RAG.
Therefore, we explore an alternative to textual serialization for conveying graph evidence to LLMs, aiming to alleviate these limitations.

In LLMs, queries and retrieved evidence ultimately interact through attention over their internal representations. 
Inspired by this mechanism, we embed query-evidence interactions into latent space and integrate it with the LLM’s internal representations, avoiding information missing in textual serialization.
To this end, we propose FuseRAG (\textbf{FU}sing \textbf{S}tructure and \textbf{E}vidence into transformer inference), a novel framework for efficient graph-based RAG,
which shift the primary evidence interface from textual serialization to latent representations.
FuseRAG consists of three stages:
In the stage 1, FuseRAG retrieves query-aware evidence from the knowledge graph with greedy heuristic search, 
preserving the retrieved subgraph structure for subsequent fusion.
In the next stage, dual-Path Evidence Fusion embeds the retrieved evidence into latent space and constructs compact explicit evidence.
Cross-attention between the query and evidence tokens produces query-specific node representations , which are further fused by graph propagation, 
obtaining an auxiliary state for each node. 
These auxiliary states are appended to the original evidence states as structural supplements.
In parallel, from the retrieved subgraph, a lightweight scoring mechanism selects a small set of facts based on their relevance to the query.
To maximize the total relevance score under a strict token budget of 32, 64, or 128 tokens,
the 0-1 Knapsack Algorithm is adopted for selection.
In Stage 3, FuseRAG projects the node representations into memory KV, which are integrated with text KV in the decoder’s self-attention. 
Node representations remain reusable across queries, while memory KV is reused throughout each answer generation.
By integrating reusable evidence representations and query-aware graph interactions into LLM attention,
FuseRAG moves beyond purely textual context construction.
This design enables faster inference and more accurate answer generation in graph-based RAG.

Our contributions are summarized as follows:
\begin{itemize}
    \item We reveal two key limitations of evidence serialization in graph-based RAG: 
    graph structure degradation and inefficient KV-cache reuse.
    Jointly optimizing serialization order to address both limitations is NP-hard, as the optimization subsumes an NP-hard ordering problem as a special case.
Moreover, the two objectives introduce additional trade-offs.
    \item We propose FuseRAG, a cost-efficient graph-based RAG framework that integrates graph evidence and structural interactions into LLM inference through latent representations, reducing dependence on textual serialization.
    By retaining only a small amount of explicit evidence as semantic anchors, FuseRAG significantly improves inference efficiency.
    
    \item We conduct comprehensive experiments on Llama-3-8B and Qwen3.5-9B across three multi-hop QA datasets, demonstrating that FuseRAG consistently improves inference efficiency while enhancing QA performance, with up to 2.7$\times$ and 2.2$\times$ inference speedups over the fastest graph-based baseline, respectively.
    
\end{itemize}
\section{Preliminaries and Motivation}
\label{sec:p & m}

\subsection{Graph-based RAG}
\begin{table}[t]
\centering
\caption{Summary of important notations.}
\label{tab:notations}
\resizebox{\columnwidth}{!}{
\begin{tabular}{ll}
\toprule
\textbf{Notation} & \textbf{Description} \\
\midrule
$\mathcal{G}=(\mathcal{V},\mathcal{E})$
    & Knowledge graph \\
$q$
    & Input query \\
$\mathcal{G}_q=(\mathcal{V}_q,\mathcal{E}_q)$
    & Retrieved subgraph for query $q$ \\
$\mathcal{B}_q$
    & Retrieved textual evidence set \\
$\pi$
    & Serialization order \\
$X_\pi$
    & Serialized evidence context \\
$C_{\mathrm{str}}(\pi)$
    & Graph structural cost \\
$R_{\mathrm{cache}}(\pi)$
    & KV-cache reuse benefit \\
\midrule
$\mathcal{G}_t$
    & Selected subgraph at retrieval step $t$ \\
$\mathcal{A}_t$
    & Candidate retrieval actions \\
$s_t(a)$
    & Heuristic score of action $a$ \\
$K$
    & Retrieval budget \\
\midrule
$H_i$
    & Token-level representation of evidence $b_i$ \\
$z_q$
    & Query representation \\
$z_i$
    & Query-aware evidence representation \\
$f_i$
    & Graph-aware structural representation \\
$M_q$
    & Latent evidence memory \\
$S_q$
    & Explicit semantic anchors \\
$\tau$
    & Explicit evidence token budget \\
\midrule
$\mathcal{L}_{\mathrm{mem}}$
    & Memory injection layers \\
$K_l^{\mathrm{mem}},V_l^{\mathrm{mem}}$
    & Memory keys and values at layer $l$ \\
$K_l^{\mathrm{text}},V_l^{\mathrm{text}}$
    & Textual keys and values at layer $l$ \\
$r$
    & Memory projection rank \\
$d_e$
    & Evidence encoder hidden dimension \\
$d_{\mathrm{KV},l}$
    & Decoder KV projection dimension \\
$L_M$
    & Number of latent memory positions \\
$\mathcal{L}_{\mathrm{QA}}$
    & Answer generation loss \\
\bottomrule
\end{tabular}
}
\end{table}
Let $\mathcal{G}=(\mathcal{V}, \mathcal{E})$ denote a knowledge graph constructed from text corpus $\mathcal{T}$, where $\mathcal{V}$ and $\mathcal{E}$ are the sets of nodes and edges.
Each node $v\in\mathcal{V}$ has a corresponding textual evidence unit $x_v$.
Given a query $q$, a graph-based RAG system retrieves a subgraph $\mathcal{G}_q=(\mathcal{V}_q,\mathcal{E}_q)$ from $\mathcal{G}$, where $\mathcal{V}_q\subseteq\mathcal{V}$ and $\mathcal{E}_q\subseteq\mathcal{E}$.
$\mathcal{G}_q$ with supplementary information enhances LLM generation.
Since LLMs take sequential textual inputs, conventional graph-based RAG methods serialize $\mathcal{G}_q$ according to an order $\pi=(\pi_1,\pi_2,\ldots,\pi_{|\mathcal{V}_q|})$.
Therefore, the textual context is written as $X_{\pi}=\operatorname{Serialize}(\mathcal{G}_q,\pi)$.
Then, LLM generates an answer $y$ with system instructions, query $q$ and evidence context $X_{\pi}$.

Naturally, a retrieved subgraph $\mathcal{G}_q$ has multiple serialization orders.
LLMs exhibit positional biases in long contexts, utilizing information near the beginning or end more effectively \cite{liu2024lost}. 
Moreover, the connected evidence in  $\mathcal{G}_q$ becomes more difficult to integrate when scattered across distant positions in the serialized context.
Consequently, although these orders preserve the same underlying graph evidence, 
different textual sequences and relevant positions lead to substantial differences in LLM inference.

\subsection{KV-Cache Reuse in LLM Inference}
Mordern decode-only LLMs employ multiple Transformer layers.
Each layer mainly consists of a self-attention module and a feed-forward network (FFN).
At $l$-th layer, hidden representations $H^{(l)}$ are projected into queries $Q^{(l)}$, keys $K^{(l)}$, and values $V^{(l)}$ with trainable weights.
Self-attention aggregates value representations according to query-key relevance, capturing contextual information among tokens.
The resulting representations are further transformed by the FFN before being passed to the next layer.

LLM inference stage contains two phases: \emph{prefill} and \emph{decoding} \cite{vaswani2017attention}.
During prefill, LLMs process the entire input context in parallel and compute key-value states for all Transformer layers, which are stored in the KV cache for subsequent decoding.
The prefill cost increases with the context length, as longer inputs require more attention computation. 
This overhead is particularly pronounced in RAG, where retrieved evidence constitutes the majority of the input context \cite{zhang2025pqcache}. 
In the decoding phase, LLMs generate tokens autoregressively.
At each step, LLM reuses the cached key-value states of previous tokens and computes only new generated token.
Beyond within request reuse, if the current input shares an identical token prefix with a previously request, the corresponding cached key-value states can be directly reused to avoid repeated prefill computation.
This mechanism significantly reduce redundant computation in LLM inference. However, KV-cache reuse requires exact prefix matching. 
In graph-based RAG, same evidence admits multiple serialization order, resulting in different token prefixes and fewer opportunities for cache reuse \cite{lin2026kvdrive}.

\subsection{Challenges in Serialization-Order Optimization}
For a serialization order $\pi$, $\operatorname{pos}_{\pi}(v)$ is the position of node $v$ in the serialized sequence.
Keeping connected evidence close in text facilitates LLMs to capture their relationships.
For effectiveness estimation, the structural cost is defined as 
$C_{\mathrm{str}}(\pi)
=
\sum_{(u,v)\in\mathcal{E}_q}
w_{uv}
\left|
\operatorname{pos}_{\pi}(u)
-
\operatorname{pos}_{\pi}(v)
\right|,    $
where $w_{uv}$ denotes the importance of relation $(u,v)$.
A smaller $C_{\mathrm{str}}(\pi)$ indicates that connected evidence is placed closer together, better preserving relational locality \cite{wei2016speedup}.
The optimization of $C_{\mathrm{str}}(\pi)$ is NP-hard by reduction from the Minimum Linear Arrangement (MinLA) problem.
\begin{tcolorbox}[
    colback=white,
    colframe=black!50,
    boxrule=0.5pt,
    arc=1mm,
    left=5pt,
    right=5pt,
    top=4pt,
    bottom=4pt
]
\begin{theorem}
Finding a serialization order that minimizes
$C_{\mathrm{str}}(\pi)$ is NP-hard.
\end{theorem}
\begin{proof}
The MinLA problem is NP-hard.
By setting $w_{uv}=1$ for all $(u,v)\in\mathcal{E}_q$,
minimizing $C_{\mathrm{str}}(\pi)$ becomes exactly MinLA.
Therefore, minimizing $C_{\mathrm{str}}(\pi)$ is NP-hard.
\end{proof}
\end{tcolorbox}
Meanwhile, KV-cache reuse benefits from serialization orders that maximize shared token prefixes across requests.
We use $\mathcal{C}$ for the set of cached input sequences and
$\operatorname{LCP}(X_\pi,X_c)$ for the longest common prefix length between the serialized context $X_\pi$ and a cached sequence $X_c\in\mathcal{C}$.
Therefore, cache-reuse benefit is $R_{\mathrm{cache}}(\pi)
=
\max_{X_c\in\mathcal{C}}
\operatorname{LCP}(X_\pi,X_c)
\label{eq:lcp},$ measuring KV-cache reuse efficiency. 

To improve both the effectiveness and efficiency, the serialization order needs to balance graph-structure preservation and KV-cache reuse. An intuitive solution is to jointly optimize these two objectives:
\begin{equation}
\pi^{*}
=
\arg\min_{\pi}
\left[
\lambda C_{\mathrm{str}}(\pi)
-
(1-\lambda)R_{\mathrm{cache}}(\pi)
\right],
\label{eq:cost}
\end{equation}
where $\lambda\in[0,1]$ controls the trade-off between structural preservation and cache reuse.
Equivalently, the goal is to identify a Pareto-efficient order that improves one objective without unnecessarily sacrificing the other.
Unfortunately, this joint optimization is computationally intractable.
Optimizing structural locality is already NP-hard. The difficulty is further compounded by KV-cache reuse, which may conflict with structural locality. Exact joint optimization requires navigating a search space under competing objectives, making it impractical for large-scale graph-based RAG.
Therefore, we move beyond the limitations of textual serialization and encode graph evidence into reusable latent representations.
The important notiations used in this paper is shown in Table \ref{tab:notations}.

\begin{figure*}[t]
    \centering
    \includegraphics[width=1\linewidth]{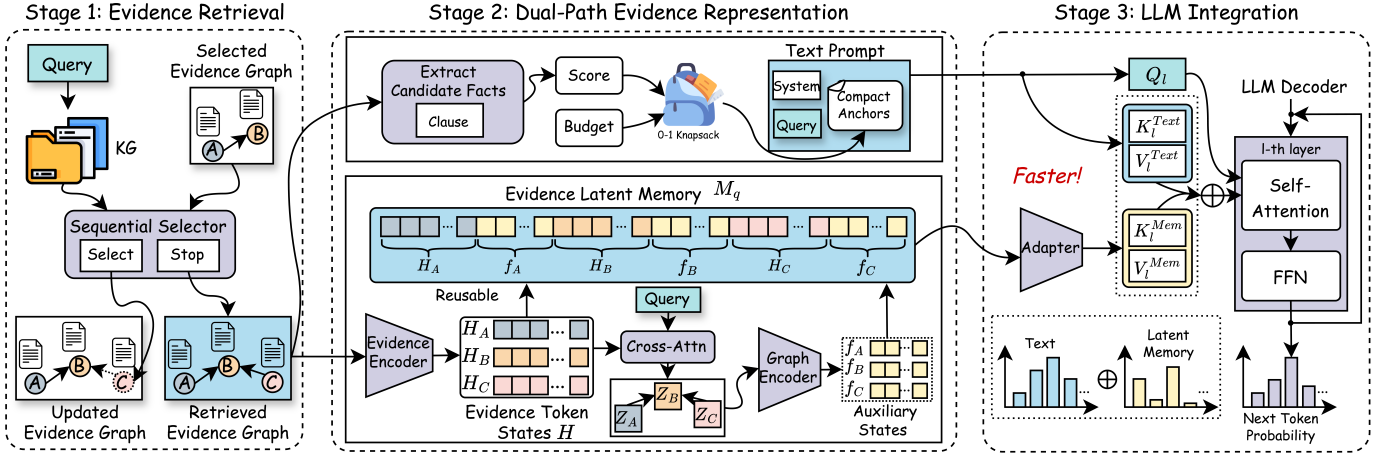}
    \caption{\textbf{Overview of FuseRAG.} The framework consists of three stages: evidence retrieval, dual-path evidence representation, and LLM integration. 
    Reusable latent evidence and graph interactions are integrated into decoder self-attention through additional memory key-value states, while compact explicit facts provide semantic anchors.}
    \label{fig:overview}
\end{figure*}

\section{Methodology}
In this section, we propose FuseRAG, which shifts the primary representation of graph evidence from textual serialization to latent representations, preserving graph structure and enabling efficient KV-cache reuse simultaneously. 
Next, we describe its design in detail.
\subsection{Overview}
Following the execution flow, FuseRAG consists of three stages: Evidence Retrieval, Dual-Path Evidence Representation, and LLM Integration:
\begin{itemize}
    \item \textbf{Stage 1: Evidence Retrieval.}
        Given a query $q$ and a knowledge graph $\mathcal{G}$, FuseRAG retrieves
    a set of relevant evidence $\mathcal{B}_q$ together with the corresponding
    subgraph $\mathcal{G}_q$.
    FuseRAG is compatible with general graph retrieval methods. 
    In our implementation, we adopt a greedy heuristic search strategy to construct \((\mathcal{B}_q,\mathcal{G}_q)\).
    Specifically, a lightweight scorer sequentially selects evidence according to the query and selected evidence.
    \item \textbf{Stage 2: Dual-Path Evidence Representation.}
    FuseRAG represents the retrieved evidence through a latent path and an explicit path.
    The latent path independently encodes each retrieved evidence unit into reusable token-level representations.
    These representations are then pass messages by graph encoder over the retrieved subgraph to capture query-relevant structural dependencies.
    The resulting latent memory $M_q$ therefore combines detailed evidence semantics with query-aware graph information without requiring the full retrieved graph to be serialized as text.
    In parallel, the explicit path extracts a compact set of query-aware facts $S_q$ under a strict token budget, providing concise semantic anchors to the LLM.
    Since only a small amount of evidence is serialized into text, the positional effects associated with long contexts are less pronounced. 

    \item \textbf{Stage 3: LLM Integration.}
    FuseRAG incorporates the two evidence representations into LLM inference.
    The latent memory $M_q$ is projected into memory key-value states with an adapter and then combined with the prompt key-value states in self-attention.
    The compact evidence $S_q$ is concatenated with the query to form the prompt.
    During decoding, the LLM attends to both the latent memory and the prompt context to compute the probability distribution over the next token.

\end{itemize}
In addition, we describe the training and inference procedures, including how the components are optimized and reused.
\subsection{Evidence Retrieval}
Although FuseRAG is compatible with common graph retrieval methods, 
we implement a lightweight selector for evidence retrieval.
Our selector constructs $\mathcal{G}_q$ by progressively expanding the currently selected subgraph, as shown in Algorithm \ref{algo:1}.
Starting from an intial node,
the selector maintains a partial subgraph $\mathcal{G}_t=(\mathcal{V}_t,\mathcal{E}_t)$ at step $t$.
Then, each candidate action $a\in\mathcal{A}_t=
\left\{
(u,r,v)\in\mathcal{E}
\,\middle|\,
u\in\mathcal{V}_t,\;
v\notin\mathcal{V}_t
\right\}
\cup \{\mathrm{STOP}\}$ corresponds to adding an unselected neighbor in knowledge graph $\mathcal{G}$ or stopping,
where $(u,r,v)$ is an edge from $u$ to $v$ with relation $r$.
Specifically, we define the score for $a$ as
\begin{equation}
s_t(a)=
\begin{cases}
\alpha R_q(v)+\beta R_q(r)+\gamma N_t(v),
& a=(u,r,v),\\[2mm]
S_{\mathrm{stop}}(q,\mathcal{G}_t),
& a=\mathrm{STOP}.
\end{cases}
\label{eq:retrieval_score}
\end{equation}
$R_q(v)$ and $R_q(r)$ measure the relevance of the candidate evidence and relation to the query computed by cosine similarity.
$N_t(v)=1-\max_{v'\in\mathcal{V}_t}
\cos(\mathbf{e}_v,\mathbf{e}_{v'})$ is the complementary information introduced by $v$ to the currently selected subgraph.
$S_{\mathrm{stop}}(q,\mathcal{G}_t)=\delta-\max_{v\in\mathcal{C}_t}S_{\mathrm{select}}(v\mid q,\mathcal{G}_t)$ estimates the maximum utility of the remaining candidates.
Retrieval stops when $S_{\mathrm{stop}}(q,\mathcal{G}_t)$ is nonnegative.

At each step, the selector greedily chooses the highest-scoring action
$a_t
=
\arg\max_{a\in\mathcal{A}_t}
s_t(a).$
For $a_t=(u,r,v)$, the selected subgraph is expanded as
$\mathcal{V}_{t+1}
=
\mathcal{V}_t\cup\{v\},
\mathcal{E}_{t+1}
=
\mathcal{E}_t\cup\{(u,r,v)\}.$
The process continues until $\mathrm{STOP}$ is selected or a predefined
retrieval budget is reached.
Then, we set
$\mathcal{V}_q=\mathcal{V}_t,
\mathcal{E}_q=
\left\{
(u,r,v)\in\mathcal{E}
\mid
u,v\in\mathcal{V}_q
\right\}.$
Both $\mathcal{B}_q$ and $\mathcal{G}_q=(\mathcal{V}_q,\mathcal{E}_q)$ are passed to the subsequent
dual-path evidence representation stage.

\subsection{Dual-Path Evidence Representation}
\subsubsection{Latent Evidence Encoding}
\label{sec:latent_encoding}

The latent path preserves fine-grained evidence semantics and
query-aware structural information without serializing the retrieved
subgraph into a long text sequence \cite{jiang2023llmlingua, liu2026autoencoding, ke2025mambev}.
For each retrieved evidence unit $b_i\in\mathcal{B}_q$, we independently
encode its text into token-level representations:
$H_i
=
E_{\phi}(b_i)
=
[h_{i,1},\ldots,h_{i,T_i}]
\in\mathbb{R}^{T_i\times d_e}$,
where $T_i$ is the number of valid tokens and $d_e$ is the hidden
dimension.
Token-level states preserve fine-grained semantic information for the decoder to access during generation.
Since $H_i$ is computed independently of the query and other evidence,
it can be cached and reused across queries.

To identify information relevant to the current query, 
cross-attention bewteen query representation $z_q$ from $q$ and $H_i$ is conducted,
$a_{i,t}=
\operatorname{Cross-Attn}_{t}
\left(
{h_{i,t}, z_q}
\right)$.
Then,
$z_i
=
\sum_{t=1}^{T_i} a_{i,t} W_v h_{i,t}$ summarizes the queay-aware evidence
$b_i$ for $q$.
Next, message passing by graph encoder over the retrieved
subgraph $\mathcal{G}_q$ aggregates information from node's neighbors: $f_i
=
\operatorname{GraphEncoder}
\left(\mathcal{G}_q,
z_i,
\{z_{j\in\mathcal{N}(i)}\}
\right)$.
The graph structure is absorbed by the auxiliary $f_i$,
capturing cross-evidence interactions.
It should be noted that, graph encoder can be instantiated with various message-passing architectures, such as GCN \cite{kipf2016semi}, GraphSAGE \cite{hamilton2017inductive}, and GAT \cite{velivckovic2017graph}.
We adopt GAT in our implementation.
Finally, the latent memory combines the token-level evidence
representations $H_i$ with the graph-aware auxiliary
representations $f_i$: 
$M_q
=
[H_1;f_1;\ldots;H_K;f_K].$

\subsubsection{Explicit Fact Anchoring}
\label{sec:explicit_anchor}
Although latent representations capture rich evidence semantics and graph structure, they lack direct textual grounding for the LLM.
FuseRAG therefore extracts a small set of query-aware facts as semantic anchors, explicitly presenting key entities and relations to guide answer generation.

For each retrieved passage $b_i\in\mathcal{B}_q$, we split it into candidate clauses $\mathcal{F}_q$.
Each clause $c\in\mathcal{F}_q$ is assigned a score that considers both query relevance and cross-evidence connections:
$r(c,q)
=
\left|
\mathcal{W}(c)\cap\mathcal{W}(q)
\right|
+
\lambda_{\mathrm{link}}\operatorname{Links}(c),$
where $\mathcal{W}(\cdot)$ is the set of normalized content words,
and
$\operatorname{Links}(c)
=
\sum_{j\neq i}
\mathbf{1}
\left[
\mathcal{W}(c)
\cap
\mathcal{W}(\operatorname{title}(b_j))
\neq\emptyset
\right]$
counts how many other retrieved passages are lexically connected to
$c$ through their titles.
We remove duplicates, and select complete clauses by solving a 0-1 knapsack problem under a strict token budget 
$\left|
\operatorname{Tok}_D(S_q)
\right|
\leq \tau$,
where $S_q$ is the resulting anchor text,
and $\operatorname{Tok}_D(S_q)$ denotes its token sequence
produced by the target decoder's tokenizer.

The resulting $S_q$ contains only a small amount of explicit evidence,
so positional effects associated with long serialized contexts are less
pronounced.
Evidence not included in $S_q$ remains available through the latent
path.
\begin{algorithm}[t]
\caption{Greedy Subgraph Retrieval}
\label{alg:retrieval}
\begin{algorithmic}[1]
\Require Query $q$, graph $\mathcal{G}=(\mathcal{V},\mathcal{E})$,
initial node $v_0$, retrieval budget $K$
\Ensure Retrieved subgraph $\mathcal{G}_q$ and evidence set $\mathcal{B}_q$

\State $\mathcal{V}_0\gets\{v_0\},\;
       \mathcal{E}_0\gets\emptyset,\;
       t\gets0$

\While{$t<K$}
    \State $\mathcal{A}_t\gets
    \{(u,r,v)\in\mathcal{E}
    \mid u\in\mathcal{V}_t,\,
    v\notin\mathcal{V}_t\}
    \cup\{\mathrm{STOP}\}$
    \State $a_t\gets
    \arg\max_{a\in\mathcal{A}_t}s_t(a)$
    
    \If{$a_t=\mathrm{STOP}$}
        \State \textbf{break}
    \EndIf
    
    \State $(u,r,v)\gets a_t$
    \State $\mathcal{V}_{t+1}\gets\mathcal{V}_t\cup\{v\},\quad
           \mathcal{E}_{t+1}\gets\mathcal{E}_t\cup\{(u,r,v)\}$
    \State $t\gets t+1$
\EndWhile

\State $\mathcal{V}_q\gets\mathcal{V}_t$, $\mathcal{E}_q\gets
\{(u,r,v)\in\mathcal{E}\mid u,v\in\mathcal{V}_q\}$
\State $\mathcal{G}_q\gets(\mathcal{V}_q,\mathcal{E}_q),\quad
       \mathcal{B}_q\gets\{\operatorname{Text}(v)\mid v\in\mathcal{V}_q\}$
\State \Return $\mathcal{G}_q,\mathcal{B}_q$
\end{algorithmic}
\label{algo:1}
\end{algorithm}
\subsection{LLM Integration}
\label{sec:llm_integration}

FuseRAG integrates latent and explicit evidence into LLM inference through complementary mechanisms. The latent memory \(M_q\) is projected into additional key-value states at selected decoder layers, enabling the LLM to access fine-grained evidence semantics and structural information through self-attention. Meanwhile, the explicit facts \(S_q\) are incorporated into the textual prompt and processed through the standard decoder path, providing direct semantic grounding for answer generation.

\subsubsection{Memory Projection}
\label{sec:memory_projection}

The latent memory $M_q\in\mathbb{R}^{L_M\times d_e}$
resides in the evidence encoder's representation space
and is not directly compatible with decoder attention.
We therefore introduce a lightweight adapter that maps
the latent representations into the decoder's key-value space \cite{houlsby2019parameter, hu2021lora}.
For each selected layer $l\in\mathcal{L}_{\mathrm{mem}}$,
the adapter employs separate low-rank projections for
memory keys and values:
$K_l^{\mathrm{mem}}
= (\operatorname{LN}(M_q)A_l^K)B_l^K,
V_l^{\mathrm{mem}}
= (\operatorname{LN}(M_q)A_l^V)B_l^V$,
where $A_l^{K/V}\in\mathbb{R}^{d_e\times r}$ and
$B_l^{K/V}\in\mathbb{R}^{r\times d_{\mathrm{KV},l}}$
are trainable projection matrices, with $r$ denoting
the projection rank.
Each decoder layer uses independent adapter parameters
to align latent evidence with its attention space.

The key and value projections use independent parameters, enabling each
decoder layer to learn how latent evidence should be matched and read in
its own attention space.
Importantly, the projection changes only the feature space while
preserving the token-level memory positions.
For a given query, the projected memory key-value states are computed once and
reused throughout autoregressive decoding.

\subsubsection{Joint Attention}
\label{sec:joint_attention}

The textual input is constructed as
$X_q=\operatorname{Template}(q,S_q)$,
and processed using the decoder's standard attention mechanism.
At an injection layer $l$, the textual and memory states are concatenated
along the key--value sequence dimension:
$\widetilde K_l
=
[K_l^{\mathrm{text}};K_l^{\mathrm{mem}}],
\widetilde V_l
=
[V_l^{\mathrm{text}};V_l^{\mathrm{mem}}].$
The attention output is then computed as
\begin{equation}
O_l
=
\operatorname{softmax}
\left(
\frac{Q_l\widetilde K_l^{\top}}
{\sqrt{d_{\mathrm{head},l}}}
+
\mathcal{A}_l
+
\mathcal{B}_l
\right)
\widetilde V_l,
\label{eq:joint_attention}
\end{equation}
where $\mathcal{A}_l$ applies the causal and memory-validity masks, and
$\mathcal{B}_l$ controls the relative preference for memory positions \cite{ke2025tinybev}.

Textual and memory states share the same softmax normalization.
Therefore, the memory keys directly influence the attention weights
assigned to their corresponding values.
The weighted memory values are then aggregated into the decoder hidden
states, which further affect the output logits and ultimately the
probability distribution over generated tokens.
As decoding proceeds, the decoder queries evolve while the latent memory
remains fixed, allowing different generation steps to dynamically attend
to different evidence positions and structural information.
\subsection{Training and Inference}
\label{sec:training_inference}

\paragraph{Training}
FuseRAG is trained to make the decoder effectively utilize both explicit
and latent evidence.
The evidence retrieval and explicit anchoring procedures are deterministic
and require no additional training.
We optimize the trainable components in the latent path, including the
adapted evidence encoder, query-conditioned graph fusion, memory
projectors, and decoder adaptation parameters.
Given a target answer
$y=(y_1,\ldots,y_T)$, we minimize the standard autoregressive
language-modeling loss:
\begin{equation}
\mathcal{L}_{\mathrm{QA}}
=
-\sum_{t=1}^{T}
\log
p_{\theta}
\left(
y_t
\mid
y_{<t},q,S_q,M_q
\right).
\label{eq:training_loss}
\end{equation}
This objective jointly trains the latent representations and their
integration into decoder attention to provide evidence useful for answer
generation.

\paragraph{Inference}
At inference time, FuseRAG first retrieves
$(\mathcal{B}_q,\mathcal{G}_q)$ and extracts the explicit anchors $S_q$.
For each retrieved evidence unit $b_i$, its cached token-level states
$H_i$ are reused when available; otherwise, $H_i$ is computed once and
inserted into the evidence cache.
FuseRAG then performs query-conditioned graph fusion to construct
$M_q$, projects it into memory key-value states, and generates the answer
autoregressively together with the textual input $X_q$.

FuseRAG therefore supports two levels of reuse.
Across queries, the query-independent evidence states $H_i$ can be
shared whenever the same evidence is retrieved.
Within a generation, the query-specific memory key-value states are computed
once and reused across decoding steps.
As decoding proceeds, only the decoder queries and textual KV cache
evolve, while the latent evidence memory remains fixed.
\section{Complexity Analysis}
\label{sec:complexity}

Let $k=|\mathcal{V}_q|$ and $m=|\mathcal{E}_q|$
denote the numbers of retrieved evidence units and edges.
The evidence contains $N=\sum_{i=1}^{k}T_i$ tokens,
and the latent memory contains $L_M=N+k$ positions.
Let $n$ and $a$ be the textual prompt length and
the number of generated tokens.
Only the explicit anchor portion of $n$ is bounded
by $\tau$.
$L_e$, $L_g$, and $L$ is the depths of the
evidence encoder, graph encoder, and decoder.
Let $d_e$ and $d$ be the evidence and decoder
representation dimensions, respectively.
The memory interface uses $J$ injection layers,
projection rank $r$, as well as key and value (KV) dimension $d_{\mathrm{KV}}$.
We assume dense Transformer computation, a fixed number of GAT heads, and a common decoder KV dimension.

\paragraph{\textbf{Stage 1: Evidence Retrieval.}}
Let $c_t$ be the number of candidate expansion
actions at step $t$, and $d_s$ be the retrieval
embedding dimension.
With precomputed node and relation embeddings,
query relevance costs $O(c_t d_s)$, and novelty
evaluation against the selected nodes costs
$O(c_t|\mathcal{V}_t|d_s)$.
Therefore,
$C_1
=
C_{\mathrm{init}}+C_{\mathrm{adj}}
+
O\left(
d_s\sum_t c_t(1+|\mathcal{V}_t|)
\right),$
where $C_{\mathrm{init}}$ is the initial-node
retrieval cost and $C_{\mathrm{adj}}$ covers
candidate enumeration and induced-subgraph construction.
With at most $K$ expansion steps and
$c_t\leq c_{\max}$, direct candidate scoring
has an upper bound of $O(c_{\max}K^2d_s)$.

\paragraph{\textbf{Stage 2: Dual-Path Evidence Representation.}}
For latent encoding, $\mathcal{U}_q$ contains
the retrieved passages without cached states.
Independent evidence encoding costs
$C_{\mathrm{enc}}
=
O\left(
L_e\sum_{i\in\mathcal{U}_q}
(T_i d_e^2+T_i^2d_e)
\right).$
This computation is avoided when all retrieved
evidence states are cached.
With dense projections, attention between query and evidence pooling costs $O(Nd_e^2)$.
Meanwhile, GAT message passing
costs $O(L_g(kd_e^2+md_e))$.
For explicit fact extraction, $F$ is the
number of candidate clauses, $W$ is their total word
count, and $W_{\mathrm{aux}}$ represents the word count of
the query and passage titles.
Word membership tests,
clause extraction, relevance scoring, cross-passage
matching, and deduplication require expected
$O(W_{\mathrm{aux}}+kW)$ computation.
Let $U_D$ denote tokenization and final budget-checking
costs.
Dynamic programming for the 0-1 knapsack problem
requires $O(F\tau)$ time.
Thus,
$C_{\mathrm{anchor}}
=
O\left(
W_{\mathrm{aux}}+kW+U_D+F\tau
\right).$
Therefore, including the query representation cost $C_q$,
the total Stage 2 cost is
$C_2={}C_{\mathrm{enc}}+C_{\mathrm{anchor}}+C_q
+O\left(
Nd_e^2+L_g(kd_e^2+md_e)
\right).$
Query-aware pooling and graph encode are
recomputed for each question.

\paragraph{\textbf{Stage 3: LLM Integration.}}
Normalization and low-rank memory adaptation
across the injection layers cost
$C_{\mathrm{proj}}
=
O\left(
JL_M[d_e+r(d_e+d_{\mathrm{KV}})]
\right).$
The projected KV is computed once per question
and reused throughout generation.
With text KV caching, standard decoder computation
for the prompt and answer costs
$C_{\mathrm{text}}
=
O\left(
L\left[
(n+a)d^2+(n^2+an+a^2)d
\right]
\right).$
At the selected layers, memory attention adds
$C_{\mathrm{mem}}
=
O\left(J(n+a)L_Md\right).$
Although memory KV is fixed, matching each new
query against memory keys and aggregating values
still incur computation.
Therefore,
$C_3
=
C_{\mathrm{proj}}+C_{\mathrm{text}}+C_{\mathrm{mem}}.$
Output vocabulary projection is not considered,
which is shared by the compared methods.

\paragraph{\textbf{Overall Cost and Efficiency Comparison.}}
The total online computation is
$C_{\mathrm{FuseRAG}}=C_1+C_2+C_3.$
For a full-text RAG baseline using the same
retrieved evidence, $p$ represents its complete
prompt length under the decoder tokenizer.
Its retrieval and decoder cost is
$C_{\mathrm{full}}
=
C_1+
O\left(
L\left[
(p+a)d^2+(p^2+ap+a^2)d
\right]
\right).$
FuseRAG reduces textual prefill attention from
$O(Lp^2d)$ to $O(Ln^2d)$ and token-wise computation
from $O(Lpd^2)$ to $O(Lnd^2)$.
Meanwhile, it introduces additional memory attention
cost of $O(J(n+a)L_Md)$ across prefill and decoding.
Unlike serialized evidence, memory positions do not
traverse the decoder's token-wise projections and
feed-forward networks at every layer.
Moreover, their evidence encodings can be amortized
across queries through caching without requiring
identical serialized prefixes.
During decoding, all $L$ layers attend to the shorter
textual context, while only $J$ layers additionally
access latent memory.
Assuming comparable per-position attention costs,
FuseRAG reduces decoding attention computation when
$L(p-n)>JL_M$.
This condition is typically achievable in FuseRAG, as latent memory is injected into only a small subset of decoder layers ($J \ll L$).

\paragraph{\textbf{Storage Cost.}}
FuseRAG requires 
$S_{\mathrm{FuseRAG}}
=
O\left(
Nd_e+JL_Md_{\mathrm{KV}}
+L(n+a)d_{\mathrm{KV}}
\right)$ excluding temporary activations,  
where $N$ is the total number of cached evidence tokens
considered in this analysis.
The three terms correspond to evidence representations,
memory KV states, and the textual KV cache, respectively.
The knapsack algorithm requires $O(\tau)$ working space,
or $O(F\tau)$ additional space when storing traceback
decisions.
In comparison, full-text RAG requires
$O(L(p+a)d_{\mathrm{KV}})$ for its decoder KV cache.
Although FuseRAG reduces textual KV-cache storage,
it introduces additional evidence and memory states.
Therefore, lower total storage is not guaranteed.
Cache lookup and data transfer may also affect
practical latency but are excluded from the
computation analysis.

\section{Experiments}
\subsection{Experimental Setup}
\subsubsection{Datasets}
we evaluate effectiveness and efficiency of FuseRAG on three public multihop Question Answering (QA) datasets,
including MusiQue \cite{trivedi2022musique}, HotpotQA \cite{yang2018hotpotqa} and 2WikiMultiHopQA \cite{ho2020constructing}: {MuSiQue} consists of questions constructed by composing single-hop questions into multi-hop reasoning chains.
{HotpotQA} contains questions that require reasoning over multiple Wikipedia articles and provides sentence to support facts.
{2WikiMultiHopQA} combines information from Wikipedia and Wikidata to construct questions with explicit reasoning paths.
These datasets are widely utilized in graph-based RAG.
Following prior work \cite{gutierrez2024hipporag,press2023measuring,trivedi2023interleaving,zhang2025sirerag,su2026clue}, 1,000 questions are sampled from the validation split of each dataset.
Meanwhile, we adopt two open-source backbone LLMs: Llama-3-8B and Qwen3.5-9B. Llama-3-8B has been widely adopted in recent RAG studies, facilitating comparison with prior work. 
Qwen3.5-9B is a newly released model that has attracted considerable attention, allowing us to further assess effectiveness.
\subsubsection{Metrics}
\paragraph{Effectiveness Metrics}
We report accuracy (Acc) and F1.
Acc measures the proportion of questions for which the generated answer contains a reference answer.
F1 measures the token-level overlap between the generated and reference answers and is computed as the harmonic mean of precision and recall.
We report both metrics averaged over the evaluated questions.

\paragraph{Efficiency Metrics}
Efficiency is measured by token consumption and runtime latency. We primarily focus online query processing, following the existing works.
Token consumption includes input and output tokens.
Runtime latency is the total time required to construct the index and the average end-to-end latency per online query.
We additionally report time to first token (TTFT), defined as the time costs from query submission until the first answer token is produced.
TTFT captures the initial response delay, while end-to-end latency reflects the total time required to complete a query.

\providecommand{\pmci}[1]{\,{\scriptsize\ensuremath{\pm\,#1}}}
\providecommand{\pmsens}[1]{\,{\scriptsize\ensuremath{\pm\,#1}\textsuperscript{S}}}

\begin{table*}[t]
\centering
\caption{QA performance and time-to-first-token (TTFT) across three datasets.}
\label{tab:main_results}
\begingroup
\fontsize{8.3}{9.8}\selectfont
\setlength{\tabcolsep}{2.7pt}
\renewcommand{\arraystretch}{1}
\setlength{\heavyrulewidth}{0.8pt}
\setlength{\lightrulewidth}{0.4pt}

\begin{adjustbox}{max totalsize={\textwidth}{0.90\textheight},center}
\begin{tabular}{@{}c@{\hspace{6pt}}l
rrr@{\hspace{9pt}}
rrr@{\hspace{9pt}}
rrr@{}}
\toprule
\multicolumn{2}{c}{}
& \multicolumn{3}{c}{\textbf{MuSiQue}}
& \multicolumn{3}{c}{\textbf{HotpotQA}}
& \multicolumn{3}{c}{\textbf{2Wiki}} \\
\cmidrule(lr){3-5}
\cmidrule(lr){6-8}
\cmidrule(lr){9-11}

\multicolumn{2}{c}{\textbf{Method}}
& \textbf{F1 (\%)}$\uparrow$
& \textbf{Acc. (\%)}$\uparrow$
& \textbf{TTFT (ms)}$\downarrow$
& \textbf{F1 (\%)}$\uparrow$
& \textbf{Acc. (\%)}$\uparrow$
& \textbf{TTFT (ms)}$\downarrow$
& \textbf{F1 (\%)}$\uparrow$
& \textbf{Acc. (\%)}$\uparrow$
& \textbf{TTFT (ms)}$\downarrow$ \\
\midrule

\multirow{20}{*}{
  \rotatebox[origin=c]{90}{\normalfont Llama-3-8B}
}
& Zero-shot
& 7.21\pmci{1.32} & 5.52\pmci{0.94} & 21.07
& 29.34\pmci{2.87} & 26.53\pmci{1.24} & 22.25
& 18.88\pmci{1.76} & 33.78\pmci{1.60} & 19.97 
\\

& VanillaRAG
& 12.28\pmci{1.14} & 14.21\pmci{1.43} & 86.23
& 40.83\pmci{2.09} & 53.87\pmci{2.78} & 89.52
& 34.22\pmci{2.54} & 41.06\pmci{2.87} & 92.75 \\
\cmidrule{2-11}
& KGP
& 8.13\pmci{0.26}  & 9.43\pmci{0.80} & 3291.84
& 22.48\pmci{0.42} & 23.57\pmci{0.67} & 3467.37
& 15.78\pmci{1.37} & 25.84\pmci{1.73} & 3753.40 \\

& ToG & 3.86\pmci{0.78} & 5.25\pmci{0.53} & 4053.57
& 3.76\pmci{0.67} & 28.79\pmci{1.76} & 4267.35
& 6.37\pmci{0.89} & 42.84\pmci{2.48}& 4553.16 \\

& RAPTOR
& 10.42\pmci{1.21} & 16.76\pmci{1.33} & 107.25
& 34.46\pmci{1.57} & 57.12\pmci{1.11} & 110.48
& 19.62\pmci{1.71} & 45.42\pmci{1.64} & 132.55 \\

& G-Retriever
& 12.25\pmci{1.28} & 14.76\pmci{1.57} & 2916.84
& 35.78\pmci{2.31} & 53.56\pmci{2.78} & 3137.67
& 24.64\pmci{2.87} & 44.58\pmci{2.41} & 3266.33 \\

& LGraphRAG
& 9.46\pmci{2.52} & 14.77\pmci{2.68} & 167.35
& 24.81\pmci{1.63} & 43.90\pmci{1.78} & 189.72
& 15.23\pmci{2.02} & 42.34\pmci{2.14} & 213.15 \\

& FastGraphRAG
& 18.24\pmci{2.53} & 17.68\pmci{2.78} & 973.79
& 29.42\pmci{2.41} & 45.76\pmci{3.21} & 1032.65
& 18.35\pmci{2.37} & 29.89\pmci{2.56} & 1075.30 \\

& LLightRAG
& 7.91\pmci{1.09} & 6.89\pmci{1.48} & 971.23
& 25.66\pmci{1.62} & 36.72\pmci{1.81} & 962.46
& 17.47\pmci{1.35} & 33.85\pmci{1.56} & 968.57 \\

& GLightRAG
& 5.25\pmci{0.77} & 5.73\pmci{0.96} & 984.69
& 17.48\pmci{1.12} & 27.72\pmci{1.49} & 979.16
& 11.67\pmci{1.56} & 23.66\pmci{1.73} & 987.35 \\

& HLightRAG
& 8.57\pmci{0.71} & 9.24\pmci{0.84} & 1013.77
& 25.93\pmci{1.31} & 38.61\pmci{1.64} & 1214.25
& 14.32\pmci{1.48} & 28.60\pmci{1.79} & 1265.32 \\

& GFM-RAG
& 29.67\pmci{1.26} & 28.64\pmci{1.34} & 2721.32
& 61.92\pmci{1.48} & 63.26\pmci{1.72} & 2836.15
& 44.28\pmci{1.75} & 45.80\pmci{1.66} & 2987.57 \\

& KET-RAG
& 24.25\pmci{2.32} & 22.51\pmci{2.88} & 1186.14
& 51.75\pmci{2.57} & 53.16\pmci{3.14} & 1212.31
& 27.90\pmci{2.79} & 41.40\pmci{3.43} & 1394.28 \\

& Clue-RAG-0.5
& 29.27\pmci{1.40} & 28.32\pmci{1.83} & 4219.46
& 55.96\pmci{1.73} & 62.08\pmci{1.56} & 4524.77
& 35.83\pmci{1.64} & 48.20\pmci{1.79} &  4667.38\\

& Clue-RAG-1.0
& 28.56\pmci{2.30} & 28.70\pmci{2.81} & 3463.22
& 56.74\pmci{2.66} & 61.95\pmci{2.76} & 3626.41
& 36.53\pmci{2.52} & 50.40\pmci{3.20} & 3789.67 \\

\cmidrule{2-11}

\rowcolor{oursblue}
\cellcolor{white}
& \textbf{Ours-32}
& 23.23\pmci{1.04} & 24.17\pmci{1.62} & \textbf{58.13}
& 51.19\pmci{1.92} & 59.38\pmci{2.71} & \textbf{44.26}
& 38.95\pmci{2.94} & 48.67\pmci{3.13} & \textbf{47.37} \\

\rowcolor{oursblue}
\cellcolor{white}
& \textbf{Ours-64}
& \underline{30.40}\pmci{1.58} & \underline{28.72}\pmci{2.03} & \underline{61.48}
& \underline{58.92}\pmci{2.84} & \underline{62.64}\pmci{3.10} & \underline{48.91}
& \underline{45.73}\pmci{3.03} & \underline{41.50}\pmci{3.10} & \underline{51.87} \\

\rowcolor{oursblue}
\cellcolor{white}
& \textbf{Ours-128}
& \textbf{41.08}\pmci{1.74} & \textbf{34.50}\pmci{2.52} & 65.74
& \textbf{63.05}\pmci{2.74} & \textbf{64.12}\pmci{3.26} & 57.13
& \textbf{55.35}\pmci{2.98} & \textbf{51.70}\pmci{3.20} & 62.25 \\

\midrule

\multirow{20}{*}{
  \rotatebox[origin=c]{90}{\normalfont Qwen3.5-9B}
}
& Zero-shot
& 7.90\pmci{0.53} & 5.94\pmci{0.77} & 77.23
& 19.77\pmci{1.02} & 20.20\pmci{1.34} & 75.37
& 23.64\pmci{1.64} & 24.62\pmci{1.60}& 62.26 \\

& VanillaRAG
& 30.71\pmci{1.24} & 25.40\pmci{1.45} & 166.54
& 59.66\pmci{2.03} & 51.30\pmci{2.31} & 176.22
& 46.12\pmci{1.79} & 48.60\pmci{1.64} & 181.67 \\
\cmidrule{2-11}
& KGP
& 23.62\pmci{1.58} & 16.91\pmci{1.12} & 4991.32
& 41.74\pmci{3.06} & 30.55\pmci{2.85} & 5433.74
& 31.61\pmci{2.88} & 30.04\pmci{2.84} & 5666.83 \\

& ToG
& 8.98\pmci{0.36} & 12.10\pmci{1.53} & 5029.18
& 8.88\pmci{0.81} & 35.64\pmci{2.97} & 5294.68
& 11.49\pmci{1.98} & 49.69\pmci{3.10} & 5655.44 \\

& RAPTOR
& 22.76\pmci{2.62} & 23.13\pmci{3.53} & 247.36
& 45.57\pmci{3.09} & 60.92\pmci{3.02} & 255.53
& 28.24\pmci{2.79} & 47.82\pmci{3.19} & 306.28 \\

& G-Retriever
& 22.87\pmci{1.74} & 22.25\pmci{1.63} & 3331.69
& 55.04\pmci{3.25} & 60.54\pmci{2.78} & 4916.75
& 40.47\pmci{3.04} & 48.78\pmci{3.07} & 4930.02 \\

& LGraphRAG
& 9.12\pmci{0.67} & 32.60\pmci{2.85} & 278.92
& 44.07\pmci{3.12} & 50.88\pmci{2.87} & 297.55
& 31.06\pmci{2.68} & 46.54\pmci{3.02} & 322.89 \\

& FastGraphRAG
& 28.86\pmci{1.53} & 25.16\pmci{1.78} & 1204.88
& 48.68\pmci{3.62} & 52.74\pmci{3.01} & 1618.57
& 34.18\pmci{2.94} & 41.49\pmci{3.05} & 1623.67 \\

& LLightRAG
& 23.04\pmci{1.09} & 19.30\pmci{1.48} & 1505.72
& 55.75\pmci{2.87} & 50.20\pmci{3.05} & 1513.67
& 46.47\pmci{1.31} & 45.50\pmci{1.54} & 1465.92 \\

& GLightRAG
& 18.53\pmci{1.47} & 16.10\pmci{2.17} & 1513.67
& 36.74\pmci{2.89} & 34.70\pmci{2.95} & 1525.41
& 31.61\pmci{1.35} & 40.10\pmci{1.48} & 1485.98 \\

& HLightRAG
& 28.86\pmci{1.22} & 25.85\pmci{1.64} & 1636.49
& 46.22\pmci{3.14} & 55.34\pmci{3.06} & 1960.91
& 34.61\pmci{2.99} & 45.21\pmci{3.03} & 2043.28 \\

& GFM-RAG
& 27.96\pmci{1.45} & 33.20\pmci{2.03}& 2676.18
& 60.21\pmci{3.03} & 67.82\pmci{2.90} & 1692.36
& 42.57\pmci{3.18} & 50.36\pmci{3.26} & 1782.18 \\

& KET-RAG
& 18.29\pmci{2.57} & 22.06\pmci{3.21} & 1398.55
& 36.53\pmci{2.98} & 31.66\pmci{2.87} & 1838.76
& 38.42\pmci{3.01} & 45.88\pmci{3.09} & 2114.42 \\

& Clue-RAG-0.5
& 39.87\pmci{3.56} & 35.80\pmci{4.27} & 4926.88
& 62.46\pmci{2.91} & 53.71\pmci{2.83} & 5075.57
& 45.26\pmci{3.00} & 50.16\pmci{3.08} & 5131.92 \\

& Clue-RAG-1.0
& 40.62\pmci{3.21} & 36.20\pmci{4.10} & 5027.59
& 60.75\pmci{2.59} & 55.40\pmci{3.10} & 4326.24
& 50.36\pmci{2.89} & 46.60\pmci{3.10} & 4487.03 \\

\cmidrule{2-11}

\rowcolor{oursblue}
\cellcolor{white}
& \textbf{Ours-32}
& 37.59\pmci{1.78} & 29.45\pmci{1.84} & \textbf{164.57}
& 41.35\pmci{2.78} & 33.60\pmci{3.00} & \textbf{141.82}
& 33.56\pmci{2.78} & 29.70\pmci{2.80} & \textbf{143.13} \\

\rowcolor{oursblue}
\cellcolor{white}
& \textbf{Ours-64}
& \underline{41.94}\pmci{2.27} & \underline{34.60}\pmci{2.52} & \underline{168.45}
& \underline{53.22}\pmci{2.82} & \underline{45.60}\pmci{3.20} & \underline{143.78}
& \underline{39.71}\pmci{2.87} & \underline{36.30}\pmci{3.00} & \underline{149.86} \\

\rowcolor{oursblue}
\cellcolor{white}
& \textbf{Ours-128}
& \textbf{44.64}\pmci{2.76} & \textbf{37.90}\pmci{2.91} & 174.52
& \textbf{64.51}\pmci{2.72} & \textbf{56.20}\pmci{3.20} & 147.49
& \textbf{51.63}\pmci{2.93} & \textbf{48.80}\pmci{3.10} & 154.77 \\

\bottomrule
\end{tabular}
\end{adjustbox}
\endgroup
\end{table*}

\begin{figure*}[t]
\centering
\captionsetup[subfigure]{font=small,skip=2pt}

\begin{subfigure}[t]{0.33\textwidth}
  \centering
  \includegraphics[width=\linewidth]{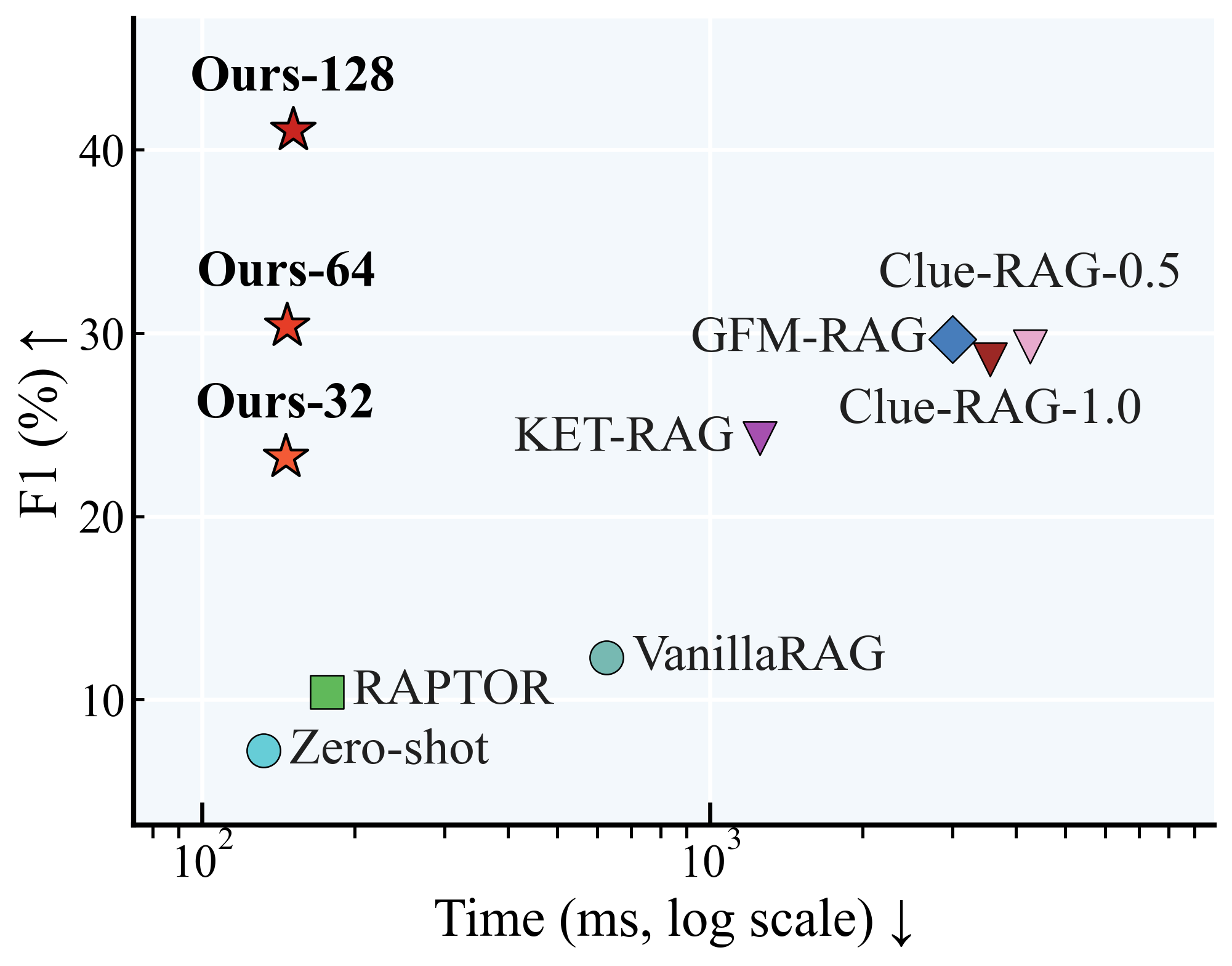}
  \caption{MuSiQue: time--F1}
\end{subfigure}\hfill
\begin{subfigure}[t]{0.33\textwidth}
  \centering
  \includegraphics[width=\linewidth]{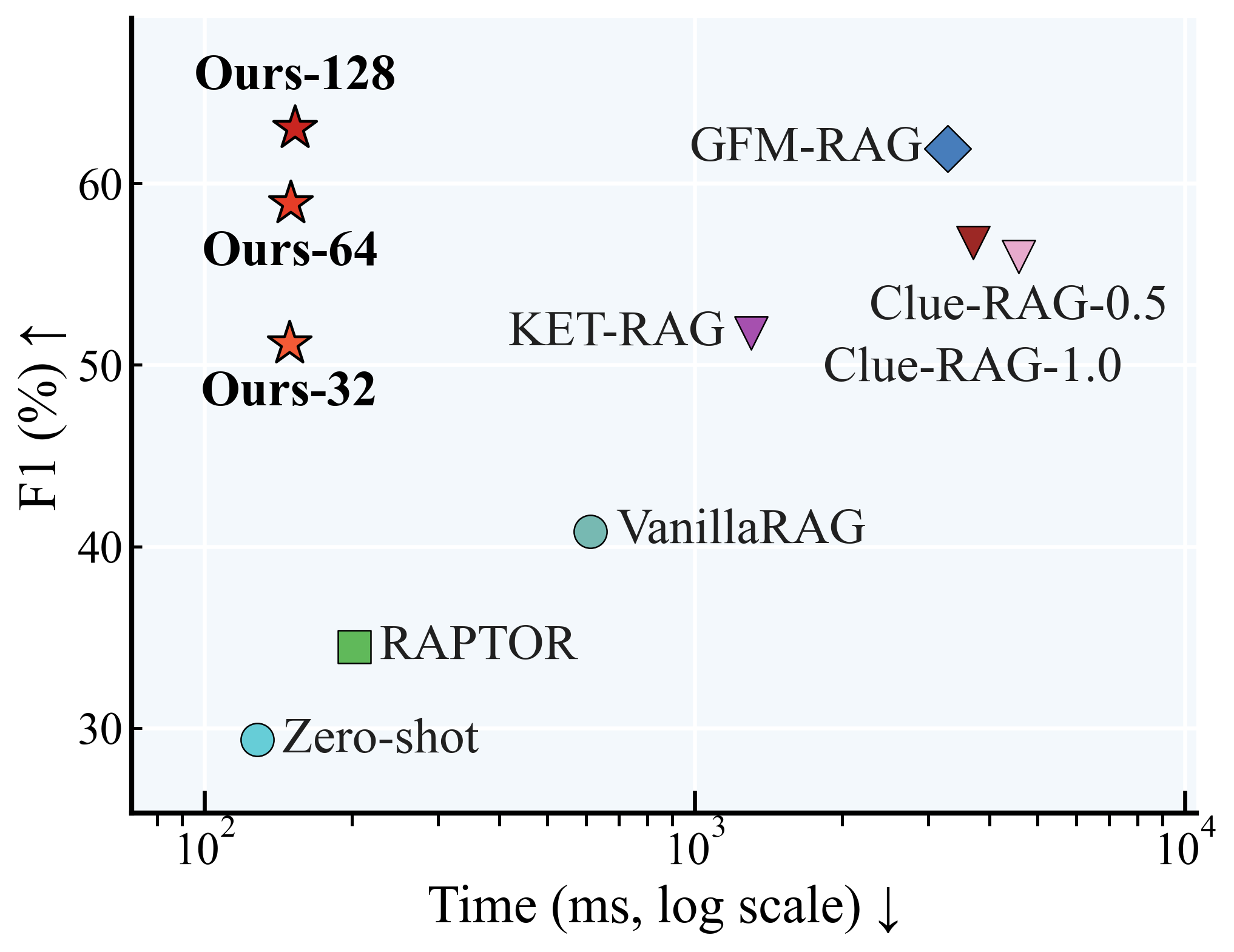}
  \caption{HotpotQA: TTFT--F1}
\end{subfigure}\hfill
\begin{subfigure}[t]{0.33\textwidth}
  \centering
  \includegraphics[width=\linewidth]{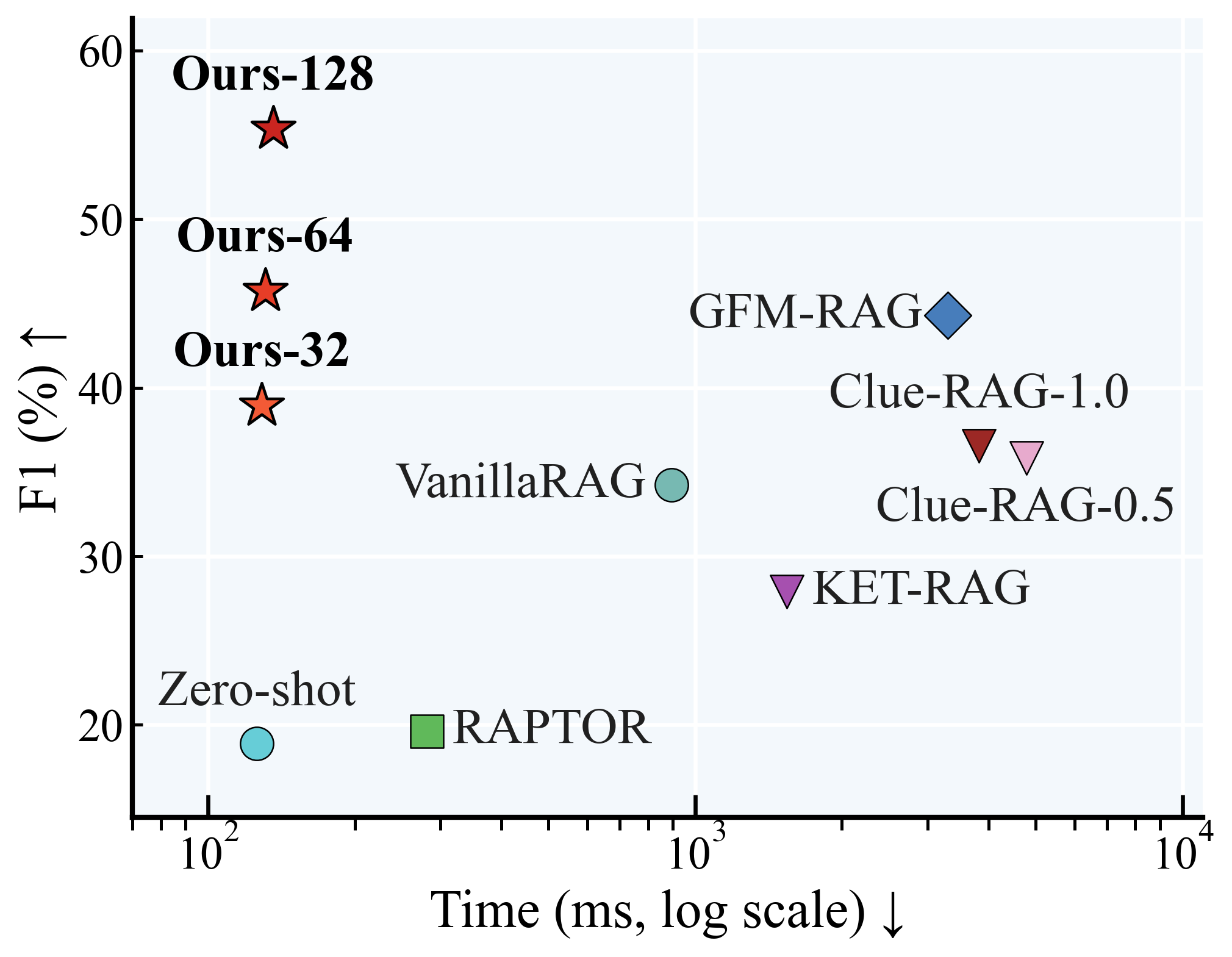}
  \caption{2Wiki: time--F1}
\end{subfigure}

\vspace{5pt}

\begin{subfigure}[t]{0.33\textwidth}
  \centering
  \includegraphics[width=\linewidth]{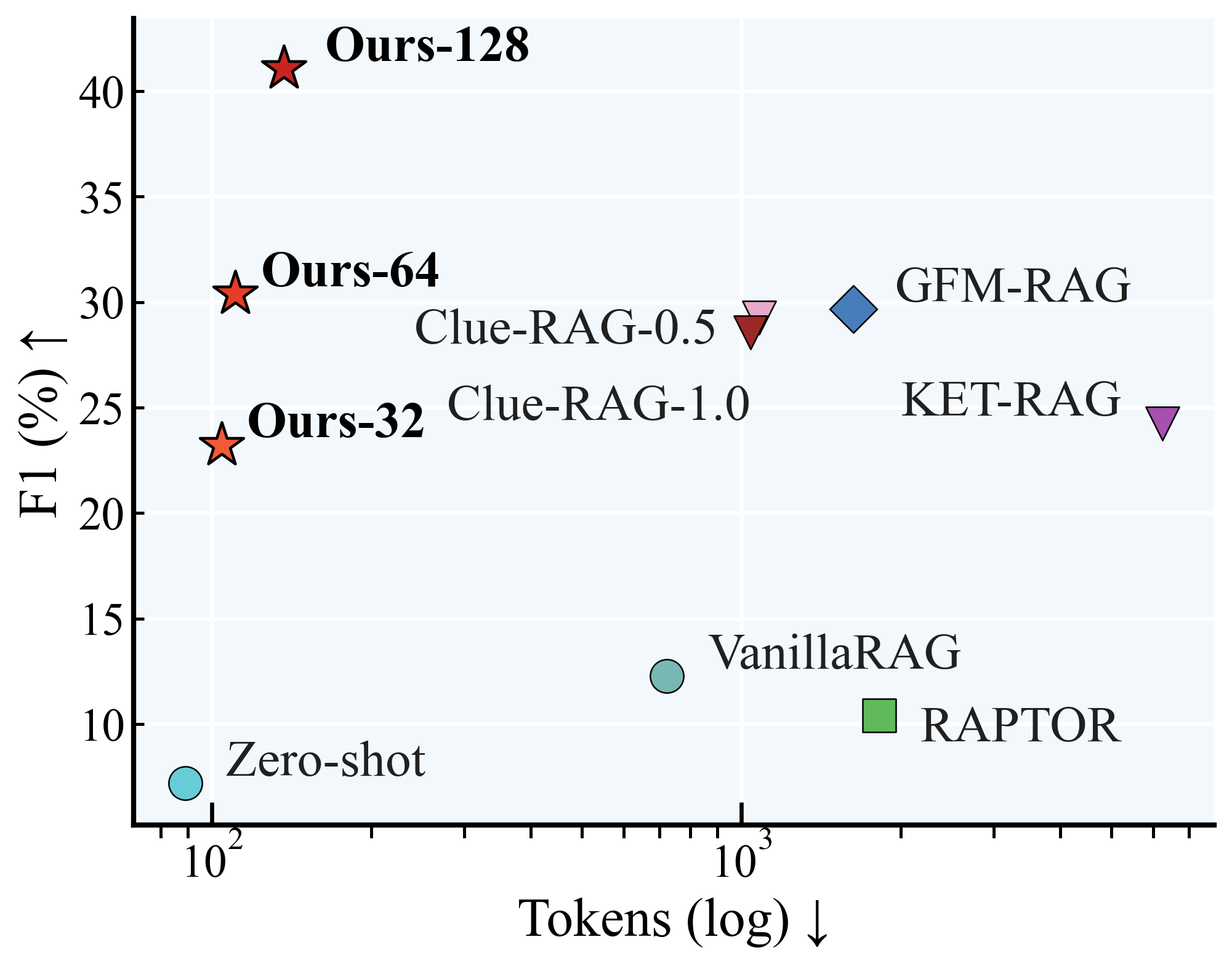}
  \caption{MuSiQue: Tokens--F1}
\end{subfigure}\hfill
\begin{subfigure}[t]{0.33\textwidth}
  \centering
  \includegraphics[width=\linewidth]{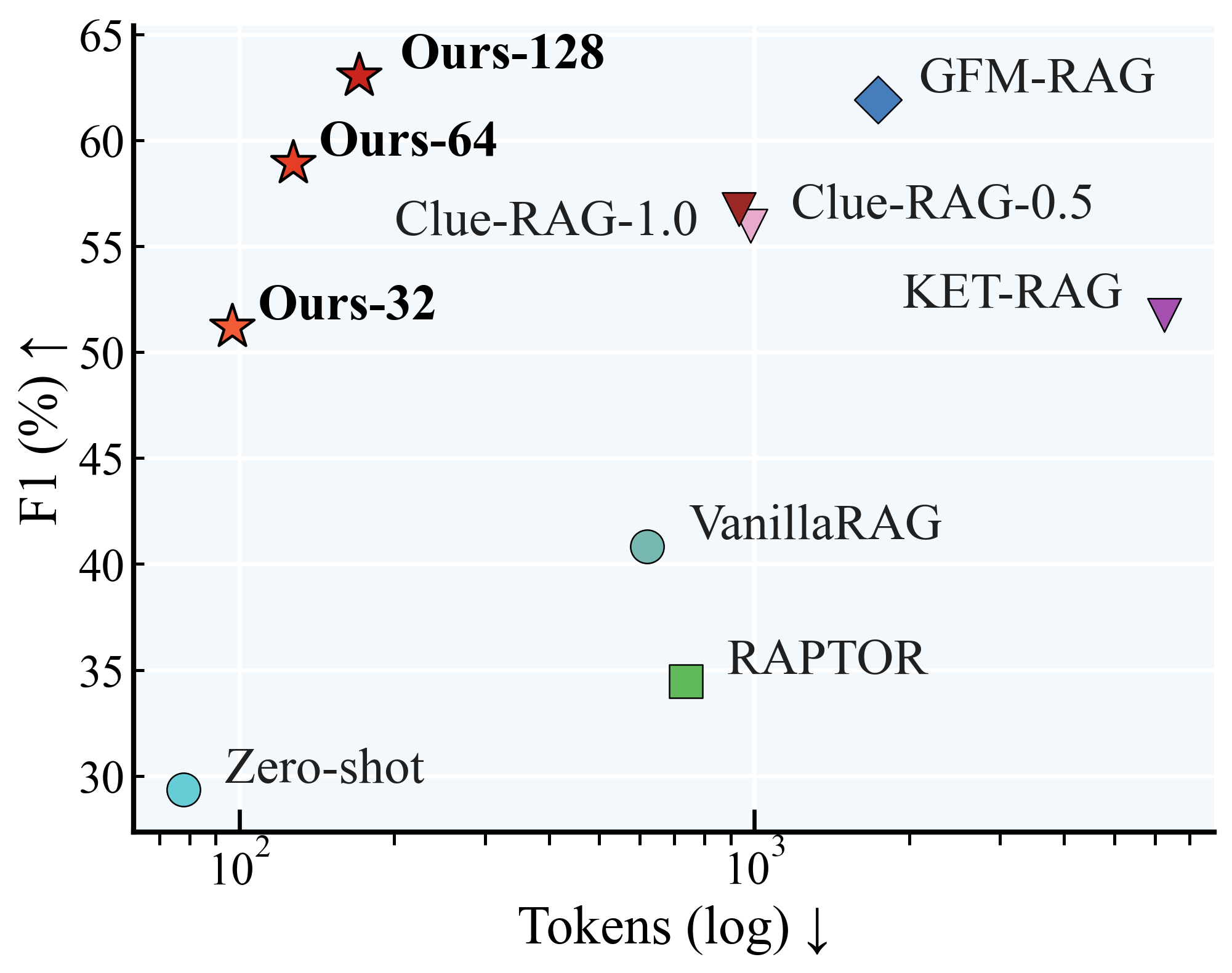}
  \caption{HotpotQA: Tokens--F1}
\end{subfigure}\hfill
\begin{subfigure}[t]{0.33\textwidth}
  \centering
  \includegraphics[width=\linewidth]{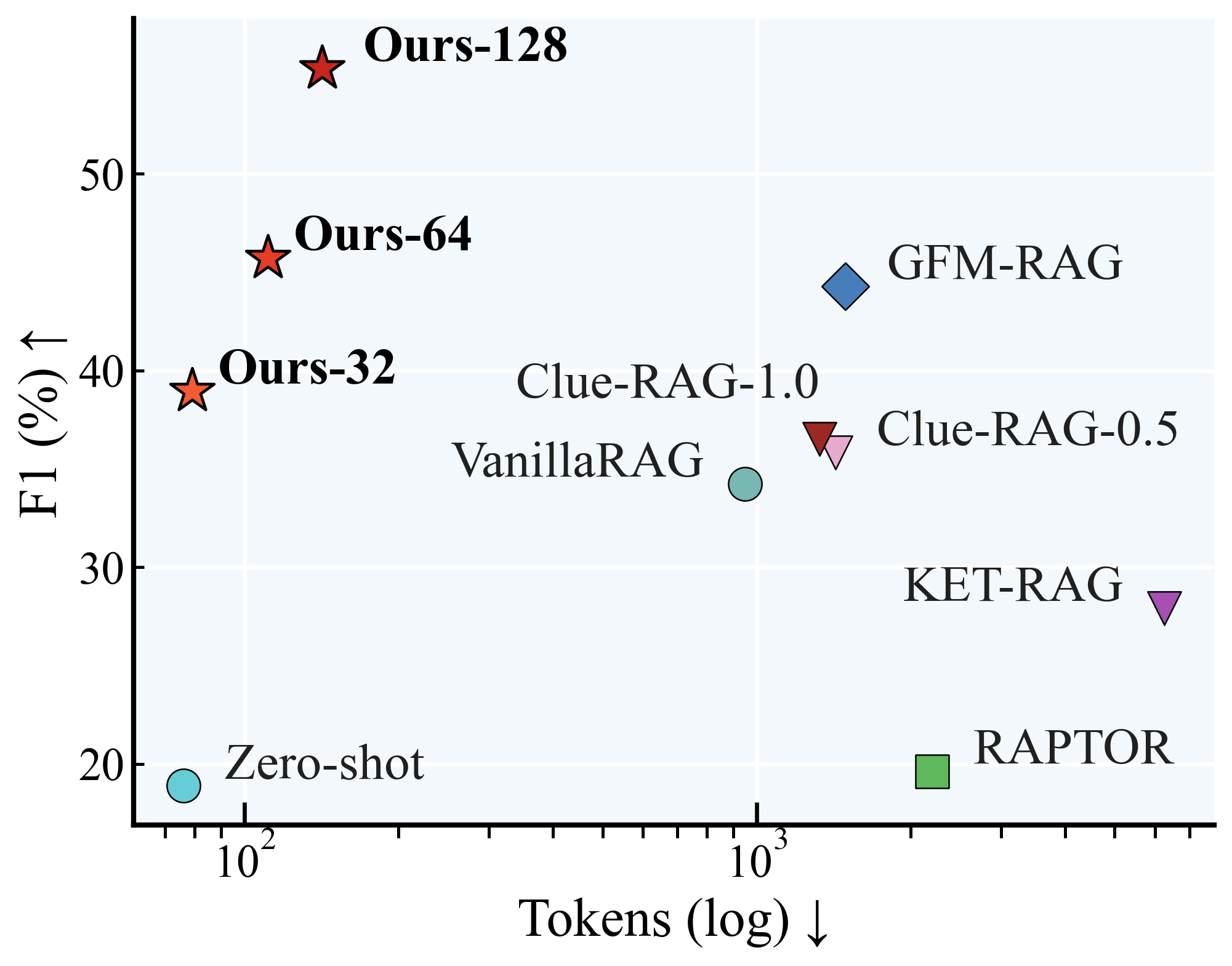}
  \caption{2Wiki: Tokens--F1}
\end{subfigure}

\caption{\textbf{QA quality versus inference cost using Llama-3-8B.}
The horizontal axes show TTFT or token cost on a logarithmic scale,
and the vertical axes show F1.
The upper-left region indicates higher quality with lower cost.
Stars denote our methods.}
\label{fig:llama_quality_cost}
\end{figure*}

\begin{figure*}[t]
    \centering
    \includegraphics[width=\textwidth]
    {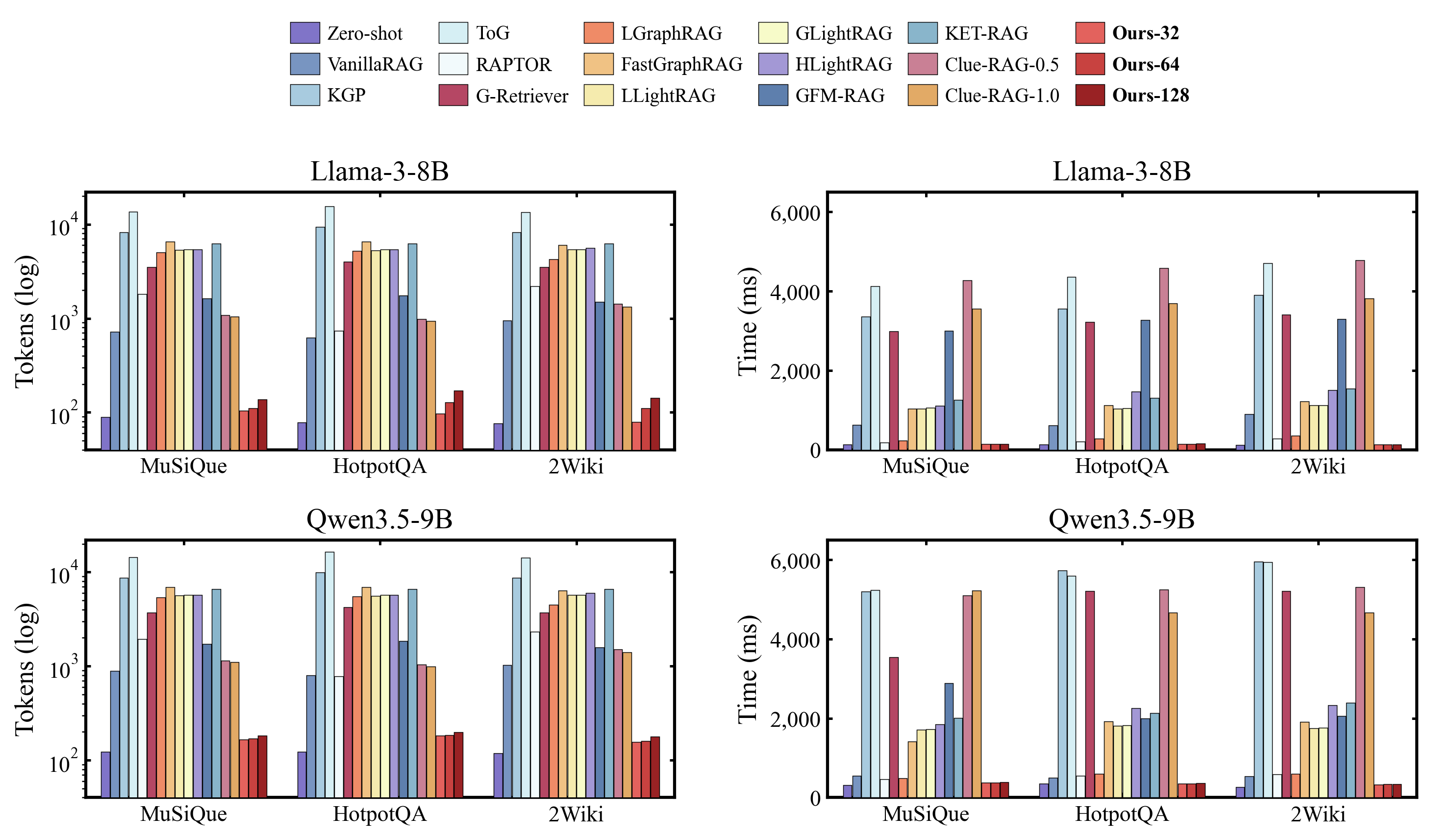}
    \caption{\textbf{Token and time costs.}}
    \label{fig:inference_efficiency}
\end{figure*}

\subsubsection{Comparison Methods}
We compare FuseRAG with two basic baselines and a range of graph-based RAG methods, covering the classic approches and the state-of-the-art \textit{(SOTA)} methods recently:

\begin{itemize}
    \item \textbf{Zero-shot.}
    The backbone LLM directly answers each question without retrieving external evidence. This baseline evaluates the model's ability to answer questions using its internal knowledge.

    \item \textbf{VanillaRAG.}
    A conventional RAG baseline retrieves text chunks based on their semantic similarity to the query and provides the retrieved chunks to the LLM for answer generation.

    \item \textbf{KGP.}
    Knowledge Graph Prompting constructs a graph over document structures, with edges capturing semantic and structural associations. An LLM-based agent navigates the graph to collect supporting passages for answer generation.

    \item \textbf{ToG.}
    An KG-based GraphRAG for online query reasoning. ToG conducts inference on manually designed KG.

    \item \textbf{RAPTOR.}
    RAPTOR recursively embeds, clusters, and summarizes text chunks to construct a hierarchical tree. It retrieves information at different levels of abstraction, combining detailed passages with higher-level summaries.

    \item \textbf{G-Retriever.}
    G-Retriever retrieves a query-relevant subgraph from a textual graph and integrates graph representations with an LLM through soft prompting for answer generation.

    \item \textbf{LGraphRAG.}
    The local search variant of GraphRAG retrieves query-relevant entities and combines their associated graph information with supporting text chunks to construct the context for answer generation.

    \item \textbf{FastGraphRAG.}
    FastGraphRAG uses personalized PageRank to explore a knowledge graph and identify relevant evidence for answering the query.

    \item \textbf{LLightRAG.}
    The local retrieval mode of LightRAG focuses on query-relevant entities and their associated information, providing context for questions about specific facts.

    \item \textbf{GLightRAG.}
    The global retrieval mode of LightRAG focuses on relationships associated with high-level query keywords, retrieving broader relational context.

    \item \textbf{HLightRAG.}
    The hybrid retrieval mode of LightRAG combines local and global retrieval results, incorporating both entity-specific information and broader relational context.

    \item \textbf{GFM-RAG.}
    GFM-RAG employs a pretrained graph neural network to model query-knowledge relationships and retrieve supporting evidence through reasoning over graph structure.

    \item \textbf{KET-RAG.}
    KET-RAG constructs a knowledge graph skeleton from selected text chunks and a lightweight text--keyword bipartite graph over the corpus. Retrieval combines evidence from both structures to balance retrieval quality and indexing cost.

    \item \textbf{Clue-RAG.}
Clue-RAG constructs a multi-partite graph linking text chunks, knowledge units, and entities, and retrieves supporting evidence through query-driven iterative retrieval. We evaluate two configurations: Clue-RAG-0.5 and Clue-RAG-1.0, with extraction budget coefficients $\alpha=0.5$ and $\alpha=1.0$, respectively. Clue-RAG-0.5 uses LLM-based knowledge unit extraction for chunks within a budget of 50\% of the corpus tokens and lightweight NLP tools for the remainder, whereas Clue-RAG-1.0 applies LLM-based extraction to all chunks.
\item \textbf{FuseRAG (ours).}
We evaluate FuseRAG with short-fact budgets of 32, 64, and 128 (Ours-32, Ours-64, and Ours-128), respectively.
Each budget specifies the maximum number of short facts provided to the LLM through the explicit path as semantic anchors.
These configurations allow us to examine how the amount of explicit factual evidence affects QA effectiveness, token consumption, and response latency.
\end{itemize}
\subsubsection{Implementation}
To ensure a fair comparison, we implement all methods within a unified experimental framework with their default configurations \cite{su2026clue}.
All experiments are conducted on two NVIDIA RTX PRO 6000 GPUs.

\subsection{Main Results}

Table~\ref{tab:main_results} reports the QA performance and TTFT of all methods.
The results demonstrate that FuseRAG achieves superiority on QA quality.

With Llama-3-8B, Ours-128 gains the highest F1 and accuracy on all three datasets.
Specifically, it obtains F1 scores of 41.08\%, 63.05\%, and 55.35\% on MuSiQue, HotpotQA, and 2WikiMultiHopQA (2Wiki), respectively.
Compared with the strongest competing F1 score on other methods, these results represent improvements of 11.41, 1.13, and 19.52\%.
Ours-128 also has an accuracy of 34.50\%, 64.12\%, and 51.70\%, improving 5.86\%, 0.86\%, and 1.30\%.
Qwen3.5-9B shows the same trend.
On MuSiQue, FuseRAG's F1 and accuracy reach 44.64\% and 37.90\%, improving over the \textit{SOTA} scores by 4.02 and 1.70\%.
The significant improvement indicates that FuseRAG is effective at integrating evidence for multi-hop question answering.

FuseRAG also has substantially lower TTFT than all other RAG methods on both backbone LLMs.
With Llama-3-8B, its TTFT ranges from 44.26 to 65.74 ms across the three datasets and budget configurations.
In comparison, RAPTOR, the fastest graph-based RAG method in these experiments, requires 107.25--132.55 ms.
With Qwen3.5-9B, FuseRAG maintains a TTFT of 141.82-174.52 ms, again below RAPTOR's 247.36-306 ms.
Notably, all FuseRAG configurations also achieve lower TTFT than VanillaRAG in the corresponding settings.
These results show that FuseRAG provides strong QA performance while keeping the initial response delay low.

We observe that FuseRAG reaches competitive QA performance even with a small explicit evidence budget. For example, on HotpotQA with Llama-3-8B, FuseRAG with 32 tokens has an F1 of 51.19\%, comparable to KET-RAG (51.75\%), while reducing TTFT from 1212.31 ms to 44.26 ms. This indicates that latent evidence effectively supplements the limited textual context.
Increasing the token budget generally improves F1, but the gains vary across datasets and backbone models. Notably, increasing the budget from 64 to 128 tokens improves F1 from 30.40\% to 41.08\% on MuSiQue with Llama-3-8B, and from 53.22\% to 64.51\% on HotpotQA with Qwen3.5-9B. However, larger budgets also increase TTFT, and Accuracy does not always improve monotonically. These observations suggest that different tasks and models require different amounts of explicit textual grounding. We further investigate this effect in Section~\ref{sec:para}. Overall, a 64-token budget strikes a balance between QA performance and inference efficiency in our experiments.

\begin{table*}[t]
\centering
\caption{Ablation study of FuseRAG.}
\label{tab:ablation}
\resizebox{\textwidth}{!}{
\begin{tabular}{lccccccccc}
\toprule
\textbf{Variant}
& \textbf{Explicit}
& \textbf{Token-level}
& \textbf{Query}
& \textbf{Graph}
& \textbf{Reuse}
& \multicolumn{2}{c}{\textbf{QA Performance}}
& \multicolumn{2}{c}{\textbf{Efficiency}} \\
\cmidrule(lr){7-8}
\cmidrule(lr){9-10}
& \textbf{Anchors}
& \textbf{Memory}
& \textbf{Conditioning}
& \textbf{Fusion}
& \textbf{Enabled}
& \textbf{F1}
& \textbf{Acc}
& \textbf{TTFT}
& \textbf{Speedup}\\
\midrule

FuseRAG
& \checkmark
& \checkmark
& \checkmark
& \checkmark
& \checkmark
& 30.40 & 28.72 & 61.48 & 1.00×\\

w/o Explicit Anchors
& $\times$
& \checkmark
& \checkmark
& \checkmark
& \checkmark
& 23.14& 23.88 & 58.13& 1.06×\\

w/o Latent Memory
& \checkmark
& $\times$
& $\times$
& $\times$
& $\times$
& 22.63 & 18.50 & 27.81& 2.21×\\

w/o Query Conditioning
& \checkmark
& \checkmark
& $\times$
& \checkmark
& \checkmark
& 23.46& 20.33& 60.81& 1.01×\\

w/o Graph Fusion
& \checkmark
& \checkmark
& \checkmark
& $\times$
& \checkmark
& 26.48 & 22.20 & 59.42& 1.03×\\

Passage-level Memory
& \checkmark
& $\times$
& \checkmark
& \checkmark
& \checkmark
& 23.91& 20.45& 46.37& 1.33×\\

w/o Evidence Reuse
& \checkmark
& \checkmark
& \checkmark
& \checkmark
& $\times$
& 30.40 & 28.72 & 81.49& 0.75×\\

\bottomrule
\end{tabular}
}
\end{table*}

\subsection{Efficiency}
\label{sec:efficiency}

To further evaluate the efficiency of FuseRAG, we compare its
token consumption and inference time with representative
graph-based RAG methods across three QA datasets.
Figure~\ref{fig:llama_quality_cost} illustrates the trade-off between
QA performance and inference cost, while
Figure~\ref{fig:inference_efficiency} reports token consumption and
end-to-end inference time on both backbone models.

As shown in Figure~\ref{fig:llama_quality_cost}, FuseRAG exhibits a strong trade-off between answer quality and inference cost.
Across all three datasets, its configurations are generally concentrated in the upper-left region, indicating competitive or higher F1 scores with lower TTFT and token consumption.
In particular, the 128-token configuration obtains the highest
F1 scores across all three datasets on Llama-3-8B, while
maintaining substantially lower inference costs than
competing methods.
Even with only 32 explicit tokens, FuseRAG delivers
competitive QA performance at significantly lower cost.
Increasing the explicit evidence budget further improves F1,
with relatively small increases in TTFT.
These results indicate that latent evidence effectively
supplements limited textual grounding.

Figure~\ref{fig:inference_efficiency} further compares textual token
consumption and end-to-end inference time across Llama-3-8B
and Qwen3.5-9B.
FuseRAG consistently reduces both costs compared with
representative graph-based RAG baselines.
Methods relying on extensive evidence serialization or
iterative LLM processing incur considerably higher overhead.
In contrast, FuseRAG retains only compact explicit facts
in the prompt while incorporating richer evidence and
graph structure through latent representations.

These improvements stem from two complementary mechanisms.
First, shorter textual contexts reduce decoder computation,
while latent memory is accessed only at selected Transformer
layers.
Second, reusable evidence representations eliminate
redundant encoding across queries without requiring identical
serialized prefixes.
Together, these mechanisms improve inference efficiency
while maintaining competitive performance on multi-hop QA.
Notably, the reduction in textual token consumption is accompanied by lower end-to-end inference time, indicating that the additional latent processing overhead does not offset the computational savings.

\subsection{Parameters}
\label{sec:para}
\begin{figure}[t]
    \centering
    \includegraphics[width=.95\columnwidth]{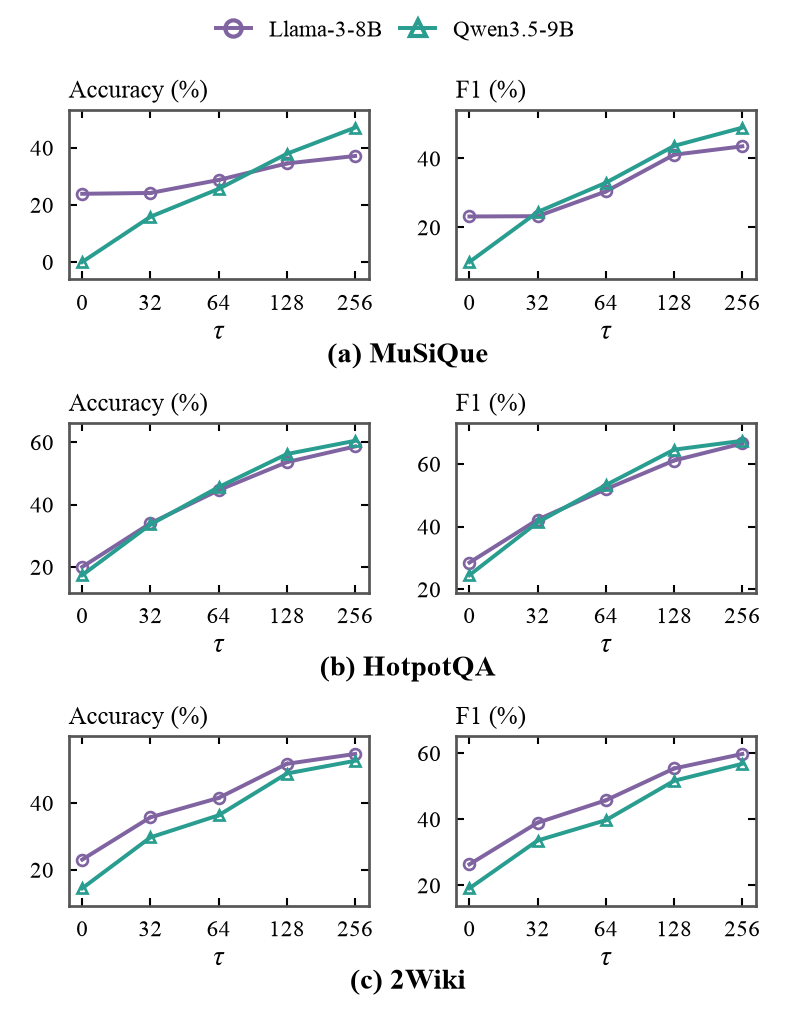}
    \caption{Answer quality across explicit token budgets $\tau$.}
    \label{fig:token-budget}
\end{figure}

\begin{figure}[t]
    \centering
    \includegraphics[width=.95\columnwidth]{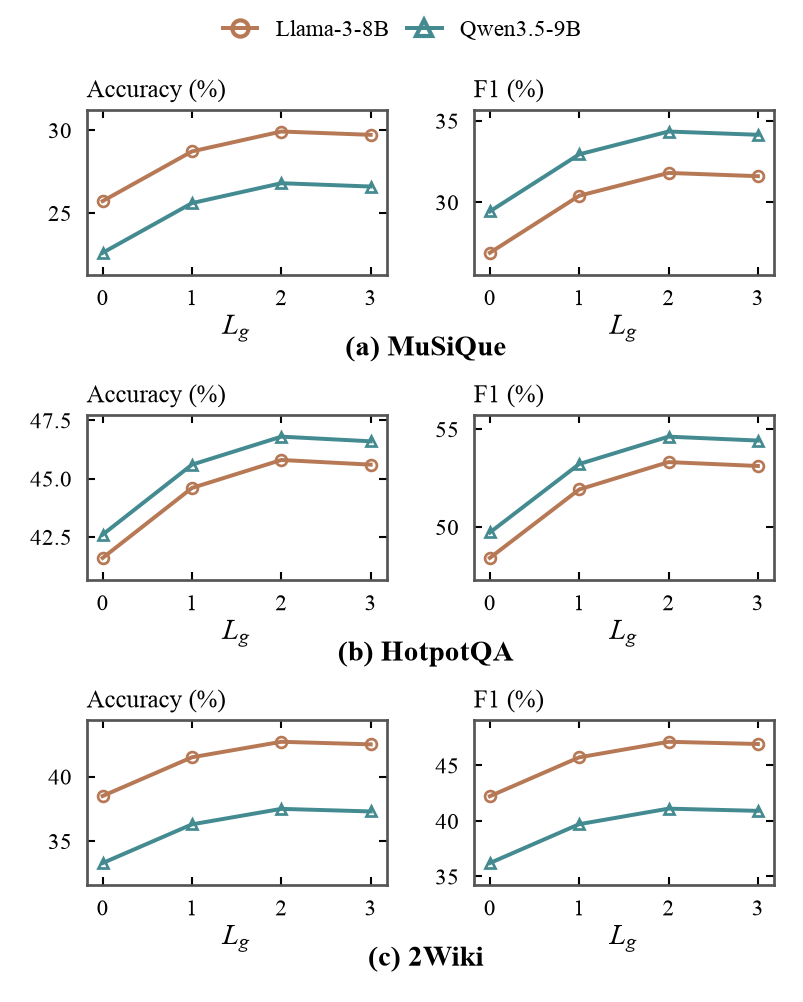}
    \caption{answer quality across graph encoder layers $L_g$.}
    \label{fig:graph-rounds}
\end{figure}
\subsection{Ablation Study}
Table~\ref{tab:ablation} examines how FuseRAG's components
contribute to answer quality and inference efficiency.
Removing explicit anchors lowers F1 from 30.40\% to 23.14\% and
accuracy from 28.72\% to 23.88\%, while providing only a
1.06$\times$ speedup. These anchors supply compact, directly
readable facts in the prompt, complementing the richer
evidence accessed through latent memory. The substantial
quality decrease suggests that latent evidence benefits from
a small amount of explicit textual context. Conversely,
removing latent memory reduces F1 to 22.63 and accuracy to
18.50, despite achieving a 2.21$\times$ speedup. This variant
retains the textual anchors but removes the latent evidence
branch, including its query conditioning and graph fusion.
The result shows that compact facts alone are insufficient
to match the full model and highlights the contribution of
the additional evidence available through memory.

Query conditioning and graph fusion serve distinct roles
within the latent branch. Query conditioning directs
attention pooling toward question-relevant evidence tokens
and conditions cross-passage message computation on the
current question. Removing it lowers F1 by 6.94\% and
accuracy by 8.39\%, with only a 1.01$\times$ speedup.
This suggests that adapting evidence representations to the
question is important for effective memory use. Graph fusion
allows connected passages to exchange information and
produces auxiliary structural representations alongside the
original evidence states. Without it, F1 decreases to 26.48
and accuracy to 22.20\%, while the speedup is limited to
1.03$\times$. Thus, structural fusion provides additional
benefits beyond independently encoded, query-conditioned
passage representations, at a relatively small observed
latency cost.

Memory granularity also affects the quality--efficiency
trade-off. Replacing token-level memory with passage-level
memory gains a 1.33$\times$ speedup, but reduces F1 to
23.91\% and accuracy to 20.45\%. Token-level memory exposes
multiple evidence positions to the decoder, enabling
selective access to finer-grained information rather than
relying on a passage-level summary. The results support
retaining this granularity when answer quality is the
priority. Finally, disabling evidence reuse leaves both
F1 and accuracy unchanged but increases TTFT from 61.48
to 81.49. Caching query-independent passage representations
therefore reduces TTFT by 24.6\% relative to online
re-encoding without changing answer quality. Taken together,
the ablations show complementary quality contributions
from textual anchors, fine-grained latent evidence, and
query-conditioned structural fusion, while evidence reuse
improves efficiency.

We examine the sensitivity of FuseRAG to the explicit evidence
token budget $\tau$ and the number of graph message-passing
rounds $L_g$ on Llama-3-8B and Qwen3.5-9B across three datasets.

Figure~\ref{fig:token-budget} varies $\tau$ from 0 to 256.
Increasing the budget generally improves both accuracy and
F1, indicating that additional explicit facts complement
latent evidence. With $\tau=0$, the models rely on the
latent evidence branch without textual anchors.
Qwen3.5-9B is particularly sensitive to this setting on
MuSiQue, whereas Llama-3-8B retains higher performance.
Adding textual anchors improves both models across the
datasets. The gains generally diminish between 128 and
256 tokens, suggesting that a moderate budget captures
much of the benefit of explicit context. Smaller budgets
provide a more compact prompt, although the corresponding
efficiency trade-off requires latency measurements.

Figure~\ref{fig:graph-rounds} varies $L_g$ from 0 to 3.
For both models, accuracy and F1 improve as the number
of rounds increases from 0 to 2 across all three datasets.
Additional rounds allow node representations to incorporate
information propagated over longer graph paths.
Performance peaks at two rounds and decreases slightly
at three, suggesting diminishing returns from further
propagation. Although excessive mixing of node information
is a possible explanation, these results alone do not
establish oversmoothing as the cause. Overall, two
message-passing rounds provide the best answer quality
among the evaluated settings.
\section{Related Work}
\subsection{Graph-based RAG}
Graph-based RAG organizes external corpora through entities, relationships, and associations, thereby capturing dependencies that go beyond isolated text blocks \cite{han2026rag, pan2024unifying}.
A majority of existing work focuses on leveraging hierarchical or relational graph structures to enhance the indexing and retrieval processes \cite{pmlr-v267-gutierrez25a, wu2026memgraphrag, pmlr-v306-cui26i}. 
For example, RAPTOR \cite{sarthi2024raptor} recursively embeds, clusters, and summarizes text blocks to construct a hierarchical retrieval tree.
GraphRAG \cite{edge2024local} and LightRAG \cite{guo2025lightrag} construct graphs and community summaries to support local, global, and hybrid retrieval modes for document collections.
To balance indexing overhead and improve multi-hop retrieval, recent methods such as HippoRAG \cite{gutierrez2024hipporag}, KET-RAG \cite{huang2025ket}, and Clue-RAG \cite{su2026clue} explore more efficient graph retrieval and indexing strategies. 
These methods adopt personalized PageRank, knowledge graph skeletons combined with text-keyword bipartite graphs, and multi-partite graph indexes connecting text chunks, knowledge units, and entities \cite{sheng2026qa}.
Another perspective focuses on performing explicit traversal or multi-step reasoning directly on graph structures during the online reasoning phase \cite{chen2026you}. 
Typical methods include KGP \cite{wang2024knowledge} using LLM agents for graph navigation,
ToG \cite{sun2024think} iterative reasoning on knowledge graphs and
GFM-RAG \cite{luo2026gfm}, which utilizes pre-trained graph neural networks to model query-knowledge relationships.

Despite their strong retrieval capabilities, most graph-based RAG methods still serialize graph evidence \cite{zhuang2026linearrag}, leading to structural degradation and positional biases. G-Retriever \cite{he2024g} mitigates this issue through graph-based soft prompts but does not explicitly preserve reusable token-level evidence. In contrast, FuseRAG separates reusable evidence content from query-dependent graph interactions and integrates both into decoder attention alongside compact textual anchors.

\subsection{Latent Evidence Integration}
Encoding external context into continuous latent representations offers a highly promising technical approach as an alternative to processing full-length text. Early research, such as Prefix-Tuning \cite{li2021prefix}, AutoCompressors \cite{chevalier2023adapting}, and ICAE \cite{ICLR2024_0b276510}, used autoencoder or language modeling objectives to compress long contexts into compact summary vectors or a fixed number of memory slots, which served as soft prompts to reduce the computational overhead of self-attention \cite{ke2026deformba}. 
In RAG scenarios, xRAG \cite{cheng2024xrag} uses a learned modality bridge to directly project pre-trained offline document embeddings into the representation space of large language models, thereby enabling the reuse of offline document embeddings and reducing the consumption of text tokens.
Relatively few studies have explored latent representations for graph-based RAG. Existing approaches often compress evidence into a few vectors, potentially losing fine-grained semantics, while overlooking graph relationships. In contrast, FuseRAG integrates evidence semantics and graph structure in latent space, supporting multi-hop reasoning without full textual serialization.

\subsection{Efficient RAG Inference}
Accelerating the prefill phase of large language models in RAG scenarios has become a critical research focus in the systems community, with mainstream approaches centering on the cross-query reuse of KV caches~\cite{liu2025kv, 10.1145/3769778, 10.14778/3819518.3819565}. 
Systems such as RAGCache~\cite{jin2025ragcache} organize intermediate knowledge states in a knowledge tree and manage cache reuse for shared prefixes across GPU and host memory. 
Since standard KV-cache reuse requires strict exact token-prefix matching, recent research has explored reuse mechanisms that overcome these prefix constraints~\cite{yang2026kvlink, fromprefix}. 
For example, CacheBlend~\cite{yao2025cacheblend} combines precomputed chunk KV caches with selective token recomputation to restore cross-chunk attention interactions between the current block and its preceding context, while Cache-Craft~\cite{agarwal2025cache} further evaluates chunk-cache reusability to jointly schedule selective recomputation, storage, and eviction strategies. 
Furthermore, recent efforts such as MatKV~\cite{shin2026matkv} and SubGCache~\cite{zhu2026subgcache} extend caching efficiency by leveraging matrix-level KV transformations for modular state alignment or exploiting overlapping subgraph structures to facilitate reuse in graph-based retrieval.

Existing caching methods primarily reuse or adapt decoder KV states derived from serialized text. However, varying graph structures and serialization orders complicate such reuse in graph-based RAG. In contrast, FuseRAG caches query-independent evidence representations and incorporates graph interactions in latent space, reducing redundant computation and prefill overhead without requiring identical token prefixes.

\section{Conclusion}
We presented FuseRAG, an efficient graph-based RAG framework that integrates retrieved evidence and structural interactions into decoder attention. Its dual-path design combines compact textual anchors with latent evidence, preserving independently encoded passage states while augmenting them with question-dependent graph representations. By reusing evidence states across queries and memory K/V across decoding steps, FuseRAG reduces repeated encoding and long-text prefill without requiring the complete retrieved graph to be verbalized. Experiments across three multi-hop QA datasets and two LLM backbones demonstrate favorable trade-offs between answer quality and inference efficiency. These results support latent evidence integration as a practical direction for efficient graph-based RAG.
\newpage
\bibliography{ref}{}
\bibliographystyle{IEEEtran}

\end{document}